\documentclass[10pt,onecolumn]{IEEEtran}
\usepackage[latin9]{inputenc}
\usepackage{amsmath,nccmath}
\usepackage{amssymb}
\usepackage{graphicx}
\usepackage{esint}
\usepackage[section]{placeins}
\usepackage{amsmath,amssymb,amsthm}
\usepackage{thmtools,thm-restate}
\usepackage{hyperref}
\usepackage{cleveref}

\usepackage{ifpdf}

\usepackage{cite}

\ifCLASSINFOpdf
 
\else
 
\fi

\usepackage{color}
\usepackage{placeins}
\usepackage{float}
\usepackage{tabularx}
\usepackage{colortbl}\usepackage{amsthm}

\usepackage{thmtools}

\begin{document}
\pdfoutput=1
\title{Search for Smart Evaders with Sweeping Agents}
%
%
%

\author{Roee M. Francos$^{1}$ and Alfred M. Bruckstein$^{1}$
\thanks{$^{1}$Roee M. Francos and Alfred M. Bruckstein are with the Faculty
of Computer Science, Technion- Israel Institute of Technology, Haifa, Israel, 32000.
        {\tt\small roee.francos@cs.technion.ac.il},
        {\tt\small alfred.bruckstein@cs.technion.ac.il}}%
}
\maketitle

\begin{abstract}
Suppose that in a given planar circular region, there are some smart mobile evaders and we would like to find them using sweeping agents. We assume that the sweeping agents are in a line formation whose total length is $2r$. We propose procedures for designing a sweeping process that ensures the successful completion of the task, thereby deriving conditions on the sweeping velocity of the linear formation and its path. Successful completion of the task means that evaders with a given limit on their velocity cannot escape the sweeping agents. A simpler task for the sweeping formation is the confinement of the evaders to their initial domain. The feasibility of completing these tasks depends on geometric and dynamic constraints that impose a lower bound on the velocity that the sweeper line formation must have. This critical velocity is derived to ensure the satisfaction of the confinement task. Increasing the velocity above the lower bound enables the agents to complete the search task as well. We present results on the total search time as a function of the sweeping velocity of the formation given the initial conditions on the size of the search region and the maximal velocity of the evaders.
\end{abstract}


%
\IEEEpeerreviewmaketitle

\section{Introduction}
An interesting challenge for multi-agent systems is the design of searching or sweeping algorithms for static or mobile targets in a region, which can either be fully mapped in advance or unknown. Often the aim is to continuously patrol a domain in order to detect intruders or to systematically search for mobile targets known to be located within some area. Search for static targets involves complete covering of the area where they are located, but a much more interesting and realistic scenario is the question of how to efficiently search for targets that are dynamic and smart. A smart target is one that detects and responds to the motions of searchers by performing optimal evasive maneuvers, to avoid interception. 

Several such problems originated in the second world war due to the need to design patrol strategies for aircraft aiming to detect ships or submarines in the English channel, see e.g. \cite{koopman1980search}. The problem of patrolling a corridor using multi agent sweeping systems in order to ensure the detection and interception of smart targets was also investigated by Vincent et. al. in \cite{vincent2004framework} and provably optimal strategies were provided by Altshuler et.al. in \cite{altshuler2008efficient}. A somewhat related, discrete version of the problem, was also investigated by Altshuler et. al. in \cite{altshuler2011multi}. It focuses on a dynamic variant of the cooperative cleaners problem, a problem that requires several simple agents to a clean a connected region on the grid with contaminated pixels. This contamination is assumed to spread to neighbors at a given rate. In \cite{mcgee2006guaranteed}, McGee et. al. investigate a search problem for smart targets. The targets do not have any maneuverability restrictions except for the maximal velocity they can move in and the sensor that the agents are equipped with detects all targets within a disk shaped area around the searcher location. The work of McGee et. al. \cite{mcgee2006guaranteed} consider search patterns consisting of spiral and linear sections. In \cite{hew2015linear}, Hew et. al. consider searching for smart evaders using concentric arc trajectories with agents sensors similar to \cite{mcgee2006guaranteed}. Such a search is proposed for detecting submarines in a channel or in a half plane. The paper focuses on determining the size of a region that can be successfully patrolled by a single searcher, where the searcher and evader velocities are known. The search problem in the paper is formulated as an optimization problem so that the search progress per arc or linear iteration, has to be maximized while guaranteeing that the evader cannot slip past the searcher undetected. 

Our work considers a scenario in which a single agent or alternatively a linear formation of several identical agents search for smart mobile targets or evaders that are to be detected. The information the agents perceive only comes from their own sensors, and evaders that intersect a sweeper's field of view are detected. We assume the single agent has a linear sensor of length $2r$ or that the linear formation of agents combine to a line shaped sensor of total length $2r$.
The evaders are initially located in a disk shaped region of radius $R_0$. There can be many evaders we wish do detect, and we consider the domain to be continuous, meaning that an evader can be located at any point in the interior of the circular region at the beginning of the search process. The sweepers are designed in a way that will require a minimal amount of memory in order to complete the required task due to the fact that the sweeping protocol is predetermined and deterministic. The sweepers move so that the line formation advances, most often perpendicularly to the agents' linear array with a speed of $V_s$ (measured at the center of the linear sensor). By assumption the evaders move at a maximal speed of $V_T$, without any maneuverability restrictions. The sweepers objective is to ``clean" or to detect all evaders that can move freely in all directions from their initial locations in the circular region of radius $R_0$. The search time will clearly depend on the type of sweeping movement the formation employs. The detection of evaders is done using a deterministic and preprogrammed circular sweeping protocol around the region. The desired result is that after each sweep around the region the radius of the circle that bounds the evader region will decrease by a value that is strictly positive. This will guarantee a complete cleaning of the evader region, by shrinking the possible area in which evaders can reside to zero, in finite time. At the beginning of the circular search process we assume that only half the formation's sensor is inside the evader region, i.e. a footprint of length $r$, while the other half is outside the region in order to catch evaders that may move outside the region while the search progresses. We analyze the proposed sweep process's performance in terms of the total time to complete the search, defined as the time at which all potential evaders that resided in the initial evader region were found. First we obtain a lower bound on the sweeper array's velocity that is independent of the search process. Then, several methods to determine the minimal velocity the linear array must have, in order to shrink the evader region to be bounded by a circle with a smaller radius than half the formation's sensor length, $r$. We then show that the minimal searcher velocity that can prevent escape cannot be based upon a single traversal of the evader region. For the case the agent or alternatively the line formation of agents, travels in a circular pattern around the evader region, we show that the minimal agent velocity that ensures satisfaction of the confinement task has to be more than twice the lower bound on a searcher velocity, and hence is not optimal. We derive two critical agent velocities that can be used together with a bisection method in order to construct an agent velocity that results in tight satisfaction of the inequality that guarantees no escape from the evader region. An analytical formula that calculates the number of required scans that are needed in order to reduce the evader region to be contained in a circle with a smaller radius than $r$, is then derived. A formula for the time it takes the agents to complete the previously mentioned scans according to the search parameters is obtained. We later show that for a line formation of sweepers equipped with line sensors a circular sweep pattern around the evader region cannot complete the cleaning of the entire area using only a circular sweeping. In order to solve the problem we provide a modification for the search process when the evader region is bounded by a circle with a radius of less than $r$. We then show that if the ratio between the searcher velocity $V_s$ and the evader's maximal velocity $V_T$ is above a certain threshold, the sweeper formation can completely clean the region performing the modified algorithm.
As opposed to our simple circular search process McGee and Hedrick in \cite{mcgee2006guaranteed} consider  more complex search patterns and use disk shaped sensors with a radius of $r$ to detect evaders. While they address the issue of the maximal region that be cleaned using their protocol they do not provide evaluations on the time it takes to find all evaders. In \cite{hew2015linear} the searcher is also assumed to have a circular sensor of radius $r$ that detects evaders if and only if they are at a distance of at most $r$ from the searcher. The goal of our work is to provide a comprehensive mathematical and geometric analysis for a simple sweeping protocol to address the sweeping detection task. This report is organized as follows: the second section provides an optimal bound on the cleaning rate that results in a minimal critical velocity a sweeper agent or alternatively a line formation of agents must have in order to successfully accomplish a confinement task. The third section provides preliminary considerations on analysis of a circular search patterns. Section $4$ provides an analysis of the critical velocity that is needed for a line formation of agents in order for it to implement a no escape search. Section $5$ provides a sweep time analysis of the proposed protocol. In the last section some extensions are discussed.

\begin{figure}[ht]
\noindent \centering{}\includegraphics[width=2.5in,height =2.5in]{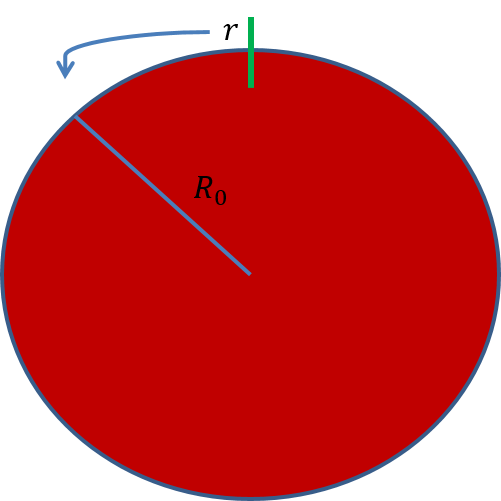} \caption{Beginning of the search using a line formation of agents.}
\label{Fig1Label}
\end{figure}

\begin{figure}[ht]
\noindent \centering{}\includegraphics[width=0.6in,height =1.2in]{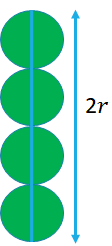} \caption{Line formation of agents with a combined sensor diameter length of $2r$. The velocity $V_s$ is measured with respect to the center of the line formation.}
\label{Fig2Label}
\end{figure}

\section{A UNIVERSAL BOUND ON CLEANING RATE}
\noindent In this section we present an optimal bound on the cleaning rate of a searcher with a linear sensor. This bound is independent of the particular search pattern employed. For each of the proposed search methods we then compare the resulting cleaning rate to the optimal derived bound in order to compare between different search methods. This will be done for the case of a single searching agent as well as for the multi-agent case. We will denote the searcher's velocity as $V_s$, the sensor length as $2r$ and the maximal velocity of an evading agent as $V_T$. The maximal cleaning rate occurs when the footprint of the sensor over the evader region is maximal. For a line shaped sensor of length $2r$ this happens when the entire length of the sensor is fully inside the evader region and it moves perpendicular to its orientation. The rate of sweeping when this happens has to be higher than the minimal expansion rate of the evader region (given its total area) otherwise no sweeping process can ensure detection of all evaders. The smallest searcher velocity satisfying this requirement is defined as the critical velocity and denoted by ${{\rm{V}}_{LB}}$, we have: 
\newtheorem{thm}{Theorem}
\begin{thm}
No sweeping process will be able to successfully complete the confinement task if its velocity, $V_s$, is less than,
$$
{{\rm{V}}_{LB}} = \frac{{\pi {R_0}{V_T}}}{r}
\eqno{(1)}
$$
\end{thm}
\begin{proof}
Denote by $\Delta T$ the interval of search. The maximal area that can be scanned when the searcher moves with a velocity $V_s$ is given by,
$$
{A_{Max\hspace{1mm} Clean}} = 2r{V_s}\Delta T
\eqno{(2)}
$$
i.e., the best cleaning rate will is $2rV_s$. The least spread of the evader region that expands due to evaders' possible motion with velocity $V_T$ occurs when the region has the shape of a circle. This is due to the isoperimetric inequality: for a given area the minimal boundary length that encloses it happens when the shape of the region is circular. Therefore for an initial circular region with radius $R_0$ the evader region minimal expansion will be to a circle with a larger radius. For a spread of $\Delta T$ the radius of the evader region can grow be $R_0 + V_T \Delta T$ and the area of the evader region will increase from $\pi {R_0}^2$ to $\pi {({R_0} + {V_T}\Delta T)^2 }$. Therefore the growth of the evader region area in time $\Delta T$ will be
${A_{Least\hspace{1mm}Spread}} = \pi {({R_0} + {V_T}\Delta T)^2 } - \pi {R_0}^2 = 2\pi {R_0}{V_T}\Delta T + {\left( {{V_T}\Delta T} \right)^2}$. The spread rate will therefore be the division of the last expression by $\Delta T$. Letting $\Delta T \to 0$ the expansion rate is $2\pi {R_0}{V_T}$, the least possible spread rate. In order to guarantee the possibility of sweeping we must set the best cleaning rate to be larger than the worst spread of area that is $2r{V_s} \geq 2\pi {R_0}{V_T}$. This yields the minimal velocity of the sweeping line regardless of the searching algorithm employed. Hence,
$$
{V_s} \geq\frac{{\pi {R_0}{V_T}}}{r}= {{\rm{V}}_{LB}}
\eqno{(3)}
$$
\end{proof}
Hopefully, after the first sweep the evader region is bounded by a circle with a smaller radius than the initial evader region's radius, and since the sweepers travel along the perimeter of the evader region and this perimeter decreases after the first sweep, ensuring a sufficient sweeper velocity that guarantees that no evader escapes during the initial sweep guarantees also that the sweeper velocity is sufficient to prevent  escape in subsequent sweeps as well.  The formulation of the problem in terms of the smallest possible searcher velocity that is needed in order to guarantee a no escape search is equivalent to asking what is the maximal boundable circular region that is possible to confine the evaders to given a searcher's velocity of $V_s$, sensor length of $2r$ and a maximal velocity of an evading agent that is equal to $V_T$.

\section{Some Preliminary Considerations on Circular Search Patterns}
\noindent In this section an intuitive and naive proposal for a cleaning search algorithm of an initial circular domain using a single agent or a line formation of agents is presented. Consider a sweeper line formation moving with a linear velocity of $V_s$ (measured at the center of the formation) and an evader with maximal linear velocity of $V_T$. It is assumed that at the beginning of the sweep process the radius of the circle bounding the evader region is $R_0$ and that the line formation of sweepers is equipped with a linear sensor whose length is $2r$. At the beginning of the circular search process half of the formation's  sensor is inside the evader region, i.e., a footprint of length $r$, while the other half is outside the region. The analysis of a desired search pattern relies on the various parameters that involve the problem such as $V_s$,$V_T$,$R_0$ and $r$. The search process that we first consider is very simple. Each time the sweeper completes a single sweep of $2\pi$ around the center of the evader region it ``steps" a positive distance towards the center of the circle and starts a new sweep. All points swept by the sensor at each round are in the sensing range of the sweeping agent, i.e. its field of view, and therefore are detected and cleaned. The search process continues in this way until the evader region is bounded by a circle of radius less than $r$ and then the sweeper employs a different form of search which will be elaborated on in the following sections. To facilitate a solution to the problem using the proposed advancement strategy, a minimal velocity of the agent that depends on the geometry of the problem as well as the evader's maximal velocity, must be maintained in order to guarantee cleaning of the domain. It is tempting to assume that it will be enough to consider a single circular traversal around the evader region in order to set a bound on the minimal searcher's velocity. Let us denote by $T$ the time it takes the sweeping line to complete a full circle around the initial evader region. This implies that $T V_s = 2\pi R_0$. At the same time the maximal distance an evader can travel in order to ensure that it is detected during the sweep by the agent is no more than $T V_T = r$. Thus, for this critical velocity of the evader we have that $\frac{2\pi R_0}{V_s} = \frac{r}{V_T}$.  This leads to the following inequality regarding the minimal, or critical linear velocity of the agent: An agent that moves at this velocity will be able to ensure that after each complete sweep the evader domain is bounded by a circle of the same radius as the initial evader region radius. This means that in order for the sweeper to be able to progress in its cleaning process it must move at this speed or above it, and hence,
$$
V_s \ge \frac{2 \pi R_0 V_T}{r}
\eqno{(4)}
$$
When we have equality in $(4)$ we denote the obtained velocity as
$$
V_{s_{\hspace{1mm} 1 \hspace{1mm} cycle}} = \frac{2 \pi R_0 V_T}{r}
\eqno{(5)}
$$
In section $5$ we will prove that this intuitive result is not correct and that the agent needs to move in a greater velocity than $V_{s_{\hspace{1mm} 1 \hspace{1mm} cycle}} =\frac{2 \pi R_0 V_T}{r}$ in order to guarantee a non escape search.

\section{Sweeping Confinement \& The Critical Velocity}
In this scenario a line formation of agents, or alternatively a single agent whose total sensing length is $2r$ are considered. A depiction of the start of the scenario is presented in Fig. $3$. The formation travels counter clockwise on the perimeter of the disk. We prove that in order for a line formation of agents to perform a non escape search their critical velocity should be based on more than one cycle in order to prevent escape outside of a circle of radius $R_0 +r$, for any evader trajectory whose maximal velocity is $V_T$. The formation employs a circular search pattern. Assuming that the sweeper formation travels counter clockwise we denote by $P$ the most problematic point in the scenario, this is the point $(0,R)$, which is just to the right of the linear sensor. An evader that spreads from point $P$ at time $t=0$, will result in an upper bound for the considered problem.  We denote by ${\chi}(\theta (t))$ the tip of the linear sensor's location as a function of the angle of the agent with respect to the center of the evader region. In order to detect all evaders we require that the envelope that describes the potential possible locations of evaders moving with a maximal velocity of $V_T$ whose origin is at time $t=0$ is at point $\varepsilon$, will always be below ${\chi}(\theta (t))$. We therefore view the evader region spread as a wave that propagates from every point in the evader region with a velocity of $V_T$. By basing our analysis on the escape from the most problematic point we guarantee that setting a searcher's velocity that ensures no escape from point $P$ ensures that there will be no escape from any other point as well.
The proof that point $P$ is the most problematic point in the cleaning of the evader region, and that this point remains the most problematic point for all cycles given that the agent scans the evader region at a fixed circular trajectory.
\begin{thm}
Point $P$ is the most problematic point in the cleaning of the evader region. This point remains the most problematic point for all cycles given that the line formation of agents scans the evader region at a fixed circular trajectory.
\end{thm}
\begin{proof}
  We denote by $\Delta {T_{1 cycle}} = \frac{{2\pi {R_0}}}{{{V_s}}}$ the time it takes the sweeper formation to complete a full cycle around the evader region. If $V_T \Delta T \leq r $ then at the end of the agents formation's traversal around the evader region the expansion of the evader region is contained in a circle of radius $R_0 +r $. At the time the line formation of agents completes the full circular traversal around the evader region, the furthest danger point from the center of the evader region is a point that originated at time $t=0$, the time the search began, from point $P$. Any other point in the evader region is closer to the center of the evader region than the furthest point at $\Delta {T_{1 cycle}}$. When the searcher reaches its original position the furthest danger points from the center are those that originated from point $P$ too. Therefore if the agent scans the evader region at a fixed circular trajectory the point $P$ is the most problematic danger point in all cycles.
\end{proof}
\newtheorem{remark}{Remark}
\begin{remark}
In order to guarantee that after each cycle the boundary will remain confined inside a circle of radius $R_0 +r$ we must look at times that correspond to a complete traversal of a cycle around the evader region in addition to a traversal of $\frac{\pi}{2}$ degrees of the searcher from the initial point and set our critical velocity based upon it.
\end{remark}
\begin{proof}
A smart evader that wishes to escape from the sweeper formation knows the center point of the evader region. This point is the point the sweeper formation sweeps around. A smart evader that starts its escape from point $\varepsilon$ will escape in the direction of the positive $y$ axis in order to increase its distance from the center of the evader region. Moving in the direction of the negative $y$ axis will reduce the distance of the evading agent from the center of evader region. Such movement is opposed to its desire to escape the region that is scanned by the formation and to its desire to move out of the region that is bounded by a circle of radius $R_0 +r$.  As the time the evader has to escape before being detected by the line formation of agents increases, it can increase its distance from the center of the evader region by choosing a suitable trajectory. Therefore since the evader is smart, it is assumed that it knows the direction of search of the scanning agent. If we assume the formation travels counter-clockwise, in order to increase the time it can escape, the evader will travel in the direction of the negative $x$ axis. From the intersection of the two constraints we conclude that the evader will escape in the directions corresponding to angles that are between $[\frac{3\pi}{2},2\pi]$. Therefore in order to set the critical velocity we need to consider search times that correspond to a full traversal around the evader region in addition to a traversal of a maximum of $\frac{\pi}{2}$ degrees from the formation's initial point.
\end{proof}
We denote by $\Delta T$ the time it takes the searcher to complete a full cycle. Therefore $\Delta T = \frac{{2\pi {R_0}}}{{{V_s}}}$ ,and the time it takes to pass a quarter of a cycle is $\widetilde T = \frac{{\pi {R_0}}}{{2{V_s}}}$. After a first circular sweep we consider the outer tip of the sensor array as it moves around, and consider the distance from it to the point $P$.
\begin{figure}[ht]
\noindent \centering{}\includegraphics[width=3.6in,height =3in]{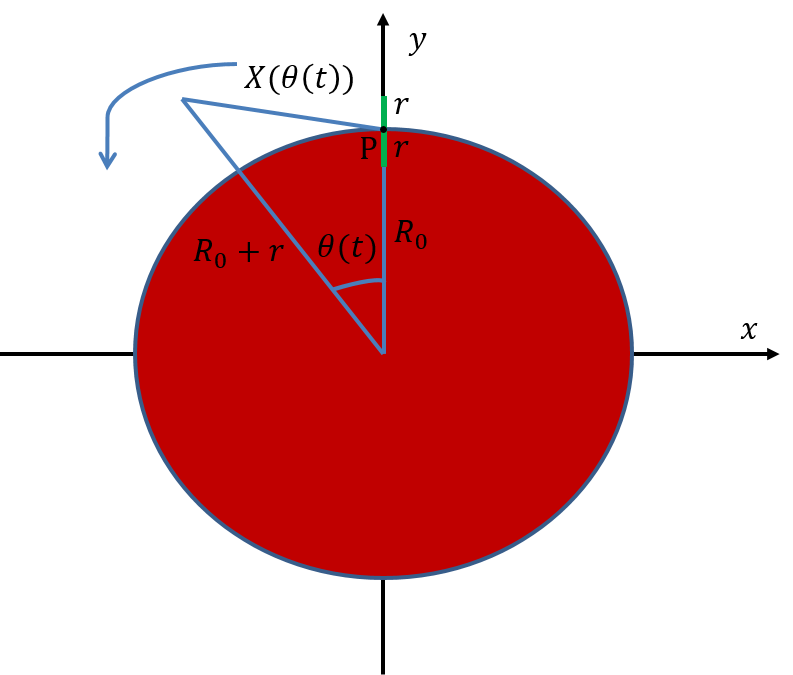} \caption{Line formation of agents with a line sensor of length $2r$. The initial radius of the evader region is $R_0$. The searchers velocity is $V_s$ and the maximal velocity of an evader is $V_T$.}
\label{Fig3Label}
\end{figure}
Applying the law of cosines for the triangle highlighted in blue results in,
$$
{\chi^2}(\theta (t)) = {(R_0 + r)^2} + {{R_0}^{2}} - 2R_0(R_0 + r)\cos \theta
\eqno{(6)}
$$
We remember that $\omega  = \dot \theta (t)$ and therefore $\theta  = \frac{{{V_s}}}{R_0}t$. To impede any possibility of escape by a smart evader from point $P$ we need to have,
$$
{\chi^2}(\theta (t)) \ge ({V_T}{\Delta _T} + {V_T}t){}^2  \hspace{1mm} \forall \hspace{1mm} t \ge 0
\eqno{(7)}
$$
This yields the requirement $\forall \hspace{1mm} t \ge 0$
$$
{({R_0} + r)^2} + {R_0}^2 - 2{R_0}({R_0} + r)\cos \left( {\frac{{{V_s}t}}{{{R_0}}}} \right)
\ge {V_T}^2{(\frac{{2\pi {R_0}}}{{{V_s}}} + t)^2}
\eqno{(8)}
$$
Rearranging terms yields the requirement to have for $V_s$ given by,
$$
\cos \left( {\frac{{{V_s}t}}{{{R_0}}}} \right) \le 1 + \frac{1}{{2{R_0}({R_0} + r)}}\left( {{r^2} - {V_T}^2\left( {\frac{{2\pi {R_0}}}{{{V_s}}} + t} \right)} \right)
\eqno{(9)}
$$
We therefore define the function $f(t,V_s)$ as,
$$
f(t,{V_s}) = 1 + \frac{1}{{2{R_0}({R_0} + r)}}\left( {{r^2} - {V_T}^2{{\left( {\frac{{2\pi {R_0}}}{{{V_S}}} + t} \right)}^2}} \right)- \cos \left( {\frac{{{V_S}t}}{{{R_0}}}} \right)
\eqno{(10)}
$$
and we wish to determine $V_s$ to satisfy,
$$
f(t,V_s) \geq 0 \hspace{3mm} \forall t \in \left[ {0,\frac{{\pi {R_0}}}{{2{V_{s}}}}} \right]
\eqno{(11)}
$$
\begin{thm}
The function $f(t,V_s)$ does not have critical minimal points when considered as a function of two variables.
$f(t,V_s)$ is a monotonically increasing function in $V_s \hspace{1mm} \forall t$.
\end{thm}
\begin{proof}
The function $f(t,V_s)$ does not have minimum points when considered as a function of two variables.
We will now prove that the function $f(t,V_s)$ is monotonically increasing in $V_s \hspace{1mm}
\forall t$. That is if $V_{s2}\geq V_{s1}$ it holds that $f(t,V_{s2}) \geq f(t,V_{s1}) \hspace{1mm} \forall t \in \left[ {0,\frac{{\pi {R_0}}}{{2{V_{s1}}}}} \right]$.  Since $\frac{{\partial f(t,{V_s})}}{{\partial {V_s}}} > 0 \hspace{1mm} \forall t \in \left[ {0,\frac{{\pi {R_0}}}{{2{V_{S1}}}}} \right]$ ,$f(t,V_s)$ cannot have minimum points when considered as a function of two variables in the domain of interest. Therefore analyzing the function along the boundary of the feasible domain, that is along the minimal $V_s$ satisfying the inequality we will search for the time at which the expression is minimal. From this minimal time we will derive a minimal value for $V_s$ that will be denoted by $V_c$ that will ensure no escape after one cycle. The minimal value for $V_s$ that will hold for the first cycle will be sufficient for all the next cycles as well, since as we shall see later, when the agent advances inwards towards the center of the evader region it traverses a circle with a smaller perimeter and hence the time it will take it to scan this perimeter will be shorter with the same $V_s$ than the time it will take it to scan the first cycle. This is true since the inequality holds for all times and therefore will allow advancing inwards after one cycle in an amount that will be analyzed later when moving in a velocity that is greater than $V_c$. In order to show that the function $f(t,V_s)$ is monotonically increasing in $V_s \hspace{1mm} \forall t \in \left[ {0,\frac{{\pi {R_0}}}{{2{V_{s}}}}} \right]$ we will show that the function's derivative is positive for all relevant times.
$$
\frac{{\partial f(t,{V_s})}}{{\partial {V_s}}} = \frac{{2\pi {V_T}^2}}{{{V_s}^2\left( {{R_0} + r} \right)}}\left( {\frac{{2\pi {R_0}}}{{{V_s}}} + t} \right) + \sin \left( {\frac{{{V_s}t}}{{{R_0}}}} \right)\frac{t}{{{R_0}}}
\eqno{(12)}
$$
The first term is always positive in the domain of interest. The sine function in the second term takes the values of $0 \le \sin \left( {\frac{{{V_s}t}}{{{R_0}}}} \right) \le 1, \hspace{1mm} \forall t \in [0,\frac{{\pi {R_0}}}{{2{V_s}}}]$ and is therefore also non negative in our domain of interest. Since the derivative of the function is the sum of a positive term and a non-negative term it is positive in the domain of interest and therefore $f(t,V_s)$ is monotonically increasing in $V_s$ for $t \in \left[ {0,\frac{{\pi {R_0}}}{{2{V_{s}}}}} \right]$.
\end{proof}
A plot that shows the behaviour of $f'(t,V_s)$ and shows that it is an increasing function can be seen in Fig. $10$, in Appendix $A$. The analysis in the next section will be done for a given $V_s$ and therefore $f(t,{V_s})$ is analyzed as a function of $t$. Since this function decreases for a very short amount of time before starting to increase, we are looking for the time $t^*$ at which the minimum value of $f(t,{V_s})$ is attained. From $(11)$ we need $f(t,V_s)$ to be positive for all times in order to guarantee no evasion. We will differentiate $f(t)$ with respect to $t$  and equate to zero in order to search for the time in which $f(t,V_s)$ reaches its minimum, and after obtaining this value of $t^*(V_s)$ we will plug it back to $(11)$ in order to solve for the minimal or critical $V_s$ that satisfies $(11)$ with equality. The times of interest, recall, are from,
$$
0 \le t \le \frac{{\pi {R_0}}}{{2{V_s}}}
\eqno{(13)}
$$
\begin{thm}
The function $f(t,V_s)$ reaches its minimum at time $t^*(V_s)$, where $t^*(V_s)$ is given by,
$$
{t^*(V_s)} = \sqrt {\frac{{ - b - \sqrt {{b^2} - 4ac} }}{{2a}}}  - \frac{{2\pi {R_0}}}{{{V_s}}}
\eqno{(14)}
$$
Where the coefficients $a,b,c$ are given by,
$$
a = {k^2}{V_T}^4,b = l - 2{k^2}{r^2}{V_T}^2 - 2k{V_T}^2,c = 2k{r^2} + {k^2}{r^4}
\eqno{(15)}
$$
With $k,l$ given by,
$$
k = \frac{1}{{2{R_0}({R_0} + r)}},l = \frac{{{V_T}^4}}{{{V_s}^2{{({R_0} + r)}^2}}}
\eqno{(16)}
$$
\end{thm}
For proof see Appendix $A$.  $f'(t,V_s)$ is an increasing function. At its zero crossing point $f(t,V_s)$ turns from descending to increasing. We are looking for the time at which this zero crossing occurs. After obtaining, $t^*$, where $f(t,V_s)$ is minimal for all values of $V_s$ we wish to find the minimal $V_s$ in which $(11)$ is satisfied with equality, that is we wish to find the value of $V_s$ in which $f(t^*,V_s) = 0 $. Plugging the values of $r=10$, $V_T = 1$ and $R_0 = 100$ results in $f(t^*,V_{s_{\hspace{1mm} 1 \hspace{1mm} cycle}}) =  - 1.047*{10^{ - 6}}<0$. This can be observed in Fig. $11$. in Appendix $A$, as well as in in Fig. $4$, where the blue plot shows $f(t^*,V_{s_{\hspace{1mm} 1 \hspace{1mm} cycle}})$. This validates our proof that there exists a set of search parameters for which $V_{s_{\hspace{1mm} 1 \hspace{1mm} cycle}}$ is not sufficient. Note that if we have $t^*(V_s)$ we can write $f(t^*(V_s),V_s)= F(V_s)$, and this expresses the minimal value of $f(t^*,V_s)$ at its critical minimal point. This derivation yields expressions that do not allow an analytical solution for $V_s$ however using numerical methods we can find a critical velocity that satisfies $\left| {f({t^*},{V_s})} \right| \le \varepsilon$, for an arbitrarily infinitesimal choice of tolerance parameter $\varepsilon$. Selecting a specific $\varepsilon$ ensures that $\left| {f({t^*},{V_{s_{\hspace{1mm} bisection}}})} \right| \le \varepsilon$ and therefore $\forall \hspace{1 mm} t ,\hspace{1 mm} {f({t},{V_{s_{\hspace{1mm} bisection}}})} + \varepsilon \geq 0$.
The proof, description of the method and informative plots that show the applicability of the method are given in Appendix $C$. If we wish to develop an analytical solution for the critical $V_s$ we will need to use some approximations, such calculations are provided in Appendix $M$. Choosing $V_{s_{\hspace{1mm} 1 \hspace{1mm} cycle}}$ will result in a value of $f(t^*,V_{s_{\hspace{1mm} 1 \hspace{1mm} cycle}})$ which is slightly less than $0$ for all choices of parameters. Another method that will be explained and derived yields a sweeper velocity, denoted by $V_{c}$  that results in a value of $f(t^*,V_c)$ which is slightly greater than $0$ for all choices of parameters. After developing $V_c$ we can apply a bisection method around $V_{s_{\hspace{1mm} 1 \hspace{1mm} cycle}}$ and $V_c$ and obtain a velocity that will result in a value of $f(t^*,V_{c})$ that is close to $0$ with any desired accuracy. Fig. $4$. shows $f(t,V_s)$ for various choices of $V_s$. It can be seen that as $V_s$ increases above a certain velocity $f(t,V_s)$ is always positive. Velocities where $f(t,V_s)\geq 0$ for all relevant times ensure a guaranteed no escape search for all possible evader trajectories satisfying a maximal evader velocity of $V_T$.

\begin{figure}[ht]
\noindent \centering{}\includegraphics[width=3.2in,height =2.8in]{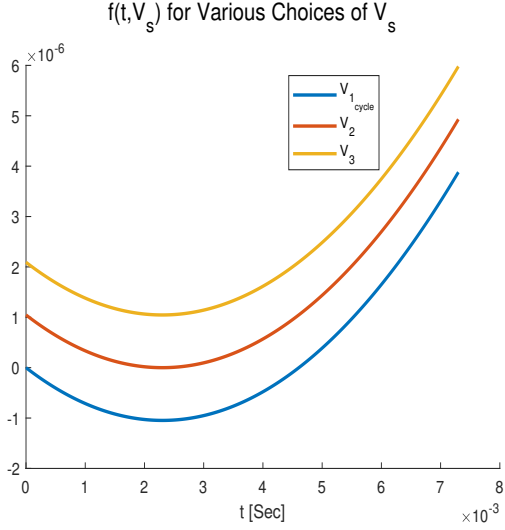} \caption{ $f(t,V_s)$ for various choices of $V_s$. It can be seen that as $V_s$ increases above a certain velocity $f(t,V_s)$ is always positive. Velocities where $f(t,V_s)\geq 0$ for all relevant times ensure a guaranteed no escape search for all possible evader trajectories satisfying a maximal evader velocity of $V_T$. The blue function represents $f(t,V_{s_{\hspace{1mm} 1 \hspace{1mm} cycle}})$. The two other plots represent values of $f(t,V_s)$ for higher values of $V_s$, namely $V_{s_{\hspace{1mm} 1 \hspace{1mm} cycle}} < V_2 < V_3$.The parameters values chosen for this plot are $r=10$, $V_T = 1$ and $R_0 = 100$.}
\label{Fig4Label}
\end{figure}

\section{The Circular Sweep Process Analysis}
Since the previously obtained solutions for the critical velocity are mathematically complicated,we propose a method to derive a slightly larger critical velocity that has a much more simple form and allows exact calculations of the search times. Previously we tried to find the tightest lower bound on the searcher's velocity by constructing a function of $2$ variables $f(t,V_s)$ demanding that the furthest possible spread of the evader region will be cleaned by the furthest tip of the formation's line sensor. A lesser requirement is to demand that by the time the most problematic point in the evader region, point $P$, spreads to a possible circle of radius of $r$ around point $P$, the searcher formation completes in addition to a full sweep around the evader region an additional angular traversal that is proportional to traversing an arc of length $r$. This means that the agent formation travels an angle of $2\pi +\beta$ where $\beta$ is marked in Fig. $5$.
\begin{figure}[ht]
\noindent \centering{}\includegraphics[width=2.4in,height =2.5in]{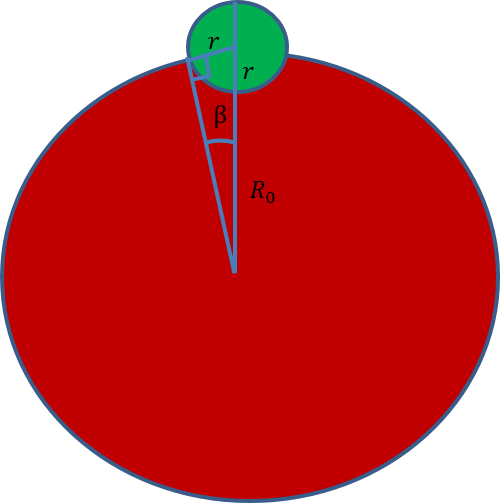} \caption{Geometric representation of the critical velocity calculation that is results in a simpler critical velocity expression. The simplified expression for the critical velocity bounds the previously found critical velocity from above for all choices of geometric parameters.}
\label{Fig5Label}
\end{figure}
We will denote the time it takes the most problematic point to spread a distance of $r$ as $T_e$. We have that $T_e = \frac{r}{{{V_T}}}$. We can see from Fig. $5$. that $\sin \beta  = \frac{r}{{R_0}}$, therefore $\beta  = \arcsin \frac{r}{{R_0}}$. The time it will takes the formation to travel an angle of $2\pi +\beta$ is therefore given by ${T_s} = \frac{{\left( {2\pi  + \arcsin \frac{r}{{R_0}}} \right){R_0}}}{{{V_s}}}$. In order to guarantee no escape we demand that $T_s \leq T_e$. Therefore rearranging terms in the previous equation and plugging $T_e$ instead of $T_s$ we get that,
$$
{V_c} \ge \frac{{\left( {2\pi  + \arcsin \frac{r}{{R_0}}} \right){R_0}{V_T}}}{r}
\eqno{(17)}
$$
The critical velocity is obtained as the lower bound of this equation. For $R_0 = 100, r=10, V_T = 1$, ${V_c} = 63.8335$.
\begin{thm}
For all search parameters satisfying that $R_0\geq r$ it holds that  $f(t^*,V_{c}) > 0$.  Therefore it holds that $\forall \hspace{1 mm} t ,\hspace{1 mm} f(t,V_{c}) \geq 0$. Thus $V_{c}$ is a sufficiently high velocity in order to accomplish the confinement task.
\end{thm}
For proof see Appendix $B$. In future derivations we use the first order Taylor approximation for the arcsine function in $(17)$, in order to enable the construction of analytical results for the sweep times of the evader region. Such an approximation is valid since in all practical scenarios the ratio between $\frac{r}{R_0}$ is sufficiently small. Therefore ${V_c}$ can be expressed as,
$$
{V_s} = \frac{{2\pi {R_0}{V_T}}}{r} + {V_T}  \eqno{(18)}
$$
Due to the mentioned geometric considerations it holds that for all search parameters satisfying that $R_0\geq r$ $f(t^*,V_{c}) > 0$.  Therefore it holds that $\forall \hspace{1 mm} t ,\hspace{1 mm} f(t,V_{c}) \geq 0$. Thus $V_{c}$ is a sufficiently high velocity in order to accomplish the confinement task. In order for the sweeper agent formation to advance inwards towards the center of the evader region it must travel in a velocity that is greater than the critical velocity. We denote by $\Delta V$ the increment in the sweeping agents velocity that is above the critical velocity. Each agent's velocity ${V_s}$ is therefore given by the sum of the critical velocity and $\Delta V$, namely
${V_s} = {V_c} + \Delta V$.
\begin{thm}
For a line formation of agents that performs the circular sweep process the number of iterations it will take the formation to reduce the evader region to be bounded by a circle with a radius that is less than or equal to $r$ is given by,
$$
N = \left\lceil {\frac{{\ln \left( {\frac{{2\pi r{V_T} - r\left( {{V_s} - {V_T}} \right)}}{{2\pi {R_0}{V_T} - r\left( {{V_s} - {V_T}} \right)}}} \right)}}{{\ln \left( {1 + \frac{{2\pi {V_T}}}{{{V_s} + {V_T}}}} \right)}}} \right\rceil 
\eqno{(19)}
$$
We denote by ${T_{in}}$ the sum of all the inward advancement times and by ${T_{circular}}$ the sum of all the circular traversal times. The time it takes the swarm to reduce the evader region to be bounded by a circle with a radius that is less than or equal to $r$ is given by, $T = {T_{circular}} + {T_{in}}$.
${T_{in}}$ is given by,
$$
{T_{in}} = \frac{{{R_0}}}{{{V_s}}} + {\left( {1 + \frac{{2\pi {V_T}}}{{{V_s} + {V_T}}}} \right)^{N - 1}}\left( {\frac{{2\pi {R_0}{V_T} - r\left( {{V_s} - {V_T}} \right)}}{{{V_s}\left( {{V_s} + {V_T}} \right)}}} \right) \eqno{(20)}
$$
And ${T_{circular}}$ is given by,
$$
\begin{array}{l}
{T_{circular}} =  - \frac{{{R_0}\left( {{V_s} + {V_T}} \right)}}{{{V_s}{V_T}}} + \frac{{r\left( {{V_s} - {V_T}} \right)\left( {{V_s} + {V_T} + 2\pi {V_T}N} \right)}}{{2\pi {V_s}{V_T}^2}}\\
 + {\left( {1 + \frac{{2\pi {V_T}}}{{{V_s} + {V_T}}}} \right)^N}\left( {\frac{{2\pi {R_0}{V_T} - r\left( {{V_s} - {V_T}} \right)}}{{{V_s}{V_T}}}} \right)\left( {\frac{{{V_s} + {V_T}}}{{2\pi {V_T}}}} \right) + \frac{{2\pi r}}{{{V_s}}}
\end{array}
\eqno{(21)}
$$
\end{thm}
\begin{proof}
  The distance a line formation of sweepers can advance inwards after completing an iteration is given by,
$$
{\delta _i}(\Delta V) = r - {V_T}\Delta {T_i} - {V_T}{T_a} = \frac{{r\left( {{V_s} - {V_T}} \right) - 2\pi {R_i}{V_T}}}{{{V_s}}}
\eqno{(22)}
$$
Where in the term ${\delta _i}(\Delta V)$, $\Delta V$ denotes the increase in the agent velocity relative to the critical velocity, and $i$ denotes the number of sweep iterations the sweeper performed around the evader region, where $i$ starts from sweep number $0$, and $\Delta T = \frac{{2\pi {R_i}}}{{{V_s}}}$. Since in $(10)$ we construct $V_s$ based on a sweeper movement of an angle of $2\pi +\beta$, and we wish that the formation will advance inwards towards the center of the evader region after a sweep of $2\pi$, the distance it can advance has to account for the additional time,denoted by ${T_a}$, that it takes it to traverse an additional angle of $\beta$. This time is given by,
$$
{T_a} = \frac{{\arcsin \frac{r}{{{R_i}}}{R_i}}}{{{V_s}}} \approx \frac{r}{{{V_s}}}
\eqno{(23)}
$$
Where the last equality results from a first order Taylor approximation of the arcsine. During ${T_a}$ the evaders continue to spread with a velocity of $V_T$ and therefore the distance the sweeper can advance decreases by $V_T{T_a}$. The time it takes the formation to move inwards until its entire sensor length is over the evader region depends on the relative velocity between the agents inwards entry and the evader region outwards expansion. Therefore the distance the formation can advance inwards after completing an iteration is given by,
$$
{\delta _{{i_{eff}}}}(\Delta V) = {\delta _i}(\Delta V)\left( {\frac{{{V_s}}}{{{V_s} + {V_T}}}} \right)
\eqno{(24)}
$$
${\delta _{{i_{eff}}}}(\Delta V)$ is a monotonically increasing function in $i$. The inwards advancement time is denoted by $T_{i{n_i}}$ and is given by,
$$
{T_{i{n_i}}} = \frac{{{\delta _{{i_{eff}}}}(\Delta V)}}{{{V_s}}} = \frac{{r\left( {{V_s} - {V_T}} \right) - 2\pi {R_i}{V_T}}}{{{V_s}\left( {{V_s} + {V_T}} \right)}}
\eqno{(25)}
$$
Where the index $i$ in ${T_{i{n_i}}}$ denotes the iteration number in which the advancement is done. Thus the new radius of the circle that will bound the evader region is given by
$$
{R_{i + 1}} = {R_i} - {\delta _{{i_{eff}}}}(\Delta V) = {R_i} - \frac{{r\left( {{V_s} - {V_T}} \right) - 2\pi {R_i}{V_T}}}{{{V_s} + {V_T}}}
\eqno{(26)}
$$
After rearranging terms we obtain and defining the coefficients $c_1$ and $c_3$ as,
$$
{c_1} =  - \frac{{r\left( {{V_s} - {V_T}} \right)}}{{{V_s} + {V_T}}},{c_3} = 1 + \frac{{2\pi {V_T}}}{{{V_s} + {V_T}}}
\eqno{(27)}
$$
$(26)$ takes the form of,
$$
{R_{i + 1}} = {c_3}{R_i} + {c_1}
\eqno{(28)}
$$
The construction of this difference equation enables to calculate the number of iterations,$N$, it takes the formation to reduce the evader region to be contained in a circle with the radius of the last scan,$\widehat{R_N} = r$. The full derivation can be found in Appendix $D$ and Appendix $E$. $N$ is given by,
$$
N = \left\lceil {\frac{{\ln \left( {\frac{{2\pi r{V_T} - r\left( {{V_s} - {V_T}} \right)}}{{2\pi {R_0}{V_T} - r\left( {{V_s} - {V_T}} \right)}}} \right)}}{{\ln \left( {1 + \frac{{2\pi {V_T}}}{{{V_s} + {V_T}}}} \right)}}} \right\rceil
\eqno{(29)}
$$
The total time it will take the formation to formation the evader region is given by total time of inward advancements combined with the times it takes the agent to complete the circular traversal of the evader region in all cycles. We denote by ${T_{in}}$ the sum of all the inward advancement times and by ${T_{circular}}$ the sum of all the circular traversal times. Namely we have that, $T = {T_{in}} + {T_{circular}}$. We denote the total advancement time until the evader region is bounded by a circle with a radius that is less than or equal to $r$ as $\widetilde{{T_{in}}}$. It is given by, $\widetilde{{T_{in}}} = \sum\limits_{i = 0}^{{N_n}-2} {{T_{i{n_i}}}}$. Since during the inward advancements only the tip of the sensor,that has zero width, is inserted into the evader region it does not clean any points until it completes its inward advance and starts sweeping again. After the formation completed its advance into the evader region its sensor footprint over the domain is equal to $r$. The total search time until the evader region is bounded by a circle with a radius that is less than or equal to $r$ is given by $\widetilde{T} = \widetilde{{T_{in}}} + \widetilde{{T_{circular}}}$. Using the developed term for $T_{i{n_i}}$ the total inward advancement times until the evader region is bounded by a circle with a radius that is less than or equal to $r$ are computed by,
$$
\widetilde{{T_{in}}} = \sum\limits_{i = 0}^{N - 2} {{T_{i{n_i}}}}  = \frac{{\left( {N - 1} \right)r\left( {{V_s} - {V_T}} \right)}}{{{V_s}\left( {{V_s} + {V_T}} \right)}} - \frac{{2\pi {V_T}\sum\limits_{i = 0}^{N - 2} {{R_i}} }}{{{V_s}\left( {{V_s} + {V_T}} \right)}}
\eqno{(30)}
$$
We note that the first inward advancement occurs when the evader region is bounded by a circle of radius $R_0$ and the last inward advancement occurs at iteration number $N-2$, which describes the inward advancement in which the evader region transitions from being bounded by a circle of radius ${R_{{N} - 2}}$ to being bounded by a circle of radius ${R_{N} - 1}$ . After iteration $N-1$ the evader region is confined to a circle with a radius that is less than or equal to $r$ and the agents perform a circular sweep in order to completely clean the evader region. The full derivation of $\widetilde{{T_{in}}}$ can be found in Appendix $J$. This derivation yields that,
$$
\widetilde{{T_{in}}} = \sum\limits_{i = 0}^{N - 2} {{T_{i{n_i}}}}  = - \frac{{r\left( {{V_s} - {V_T}} \right)}}{{2\pi {V_T}{V_s}}} + \frac{{{R_0}}}{{{V_s}}}- {\left( {1 + \frac{{2\pi {V_T}}}{{{V_s} + {V_T}}}} \right)^{N - 1}}\left( {\frac{{2\pi {R_0}{V_T} - r\left( {{V_s} - {V_T}} \right)}}{{2\pi {V_T}{V_s}}}} \right)
\eqno{(31)}
$$
In order to calculate ${T_{in}}$ we must add the last inward advancement. This time is given by ${T_{_{in}last}}=\frac{{{R_{{N}}}}}{{{V_s}}}$ and is developed in Appendix $L$. ${T_{in}}$ is given as ${T_{in}} = \widetilde{{T_{in}}} + {T_{_{in}last}}$ and therefore yields,
$$
{T_{in}} = \frac{{{R_0}}}{{{V_s}}} + {\left( {1 + \frac{{2\pi {V_T}}}{{{V_s} + {V_T}}}} \right)^{N - 1}}\left( {\frac{{2\pi {R_0}{V_T} - r\left( {{V_s} - {V_T}} \right)}}{{{V_s}\left( {{V_s} + {V_T}} \right)}}} \right)
\eqno{(32)}
$$
The total circular traversal times are computed from by multiplying the recursive radii equation ${R_{i + 1}} = {c_3}{R_i} + {c_1}$ on both sides by $\frac{2\pi}{{V_s}}$ and developing recursive formulas for the sweep times. The formulas are proved in Appendices $G$ and $H$. The initial circular traversal time is given by, ${T_0} = \frac{{2\pi {R_0}}}{{{V_s}}}$.
The last circular traversal time before the evader region is bounded by a circle with a radius that is smaller or equal to $r$, denoted by ${R_N}$, it is developed in a similar manner to the derivation of ${R_{N-2}}$ in Appendix $F$. $\widetilde{{T_{circular}}}$ is given by,
$$
{T_{circular}} =  - \frac{{{R_0}\left( {{V_s} + {V_T}} \right)}}{{{V_s}{V_T}}} + \frac{{r\left( {{V_s} - {V_T}} \right)\left( {{V_s} + {V_T} + 2\pi {V_T}N} \right)}}{{2\pi {V_s}{V_T}^2}} + {\left( {1 + \frac{{2\pi {V_T}}}{{{V_s} + {V_T}}}} \right)^N}\left( {\frac{{2\pi {R_0}{V_T} - r\left( {{V_s} - {V_T}} \right)}}{{{V_s}{V_T}}}} \right)\left( {\frac{{{V_s} + {V_T}}}{{2\pi {V_T}}}} \right)
\eqno{(33)}
$$
The last circular sweep occurs after the sweeper formation advances towards the center of the evader region and places the lower tip of its sensor at the center of the evader region. The last sweep is therefore a circular sweep with a radius of $r$. The time it takes the sweepers to complete it is given by, ${T_{last}} = \frac{{2\pi r}}{{{V_s}}}$. Therefore the total time of circular sweeps until complete cleaning of the evader region is given by,
$$
\begin{array}{*{20}{l}}
{{T_{circular}} = {T_{circular}} + {T_{last}} =  - \frac{{{R_0}\left( {{V_s} + {V_T}} \right)}}{{{V_s}{V_T}}} + \frac{{r\left( {{V_s} - {V_T}} \right)\left( {{V_s} + {V_T} + 2\pi {V_T}N} \right)}}{{2\pi {V_s}{V_T}^2}} + \frac{{2\pi r}}{{{V_s}}}}\\
{ + {{\left( {1 + \frac{{2\pi {V_T}}}{{{V_s} + {V_T}}}} \right)}^N}\left( {\frac{{2\pi {R_0}{V_T} - r\left( {{V_s} - {V_T}} \right)}}{{{V_s}{V_T}}}} \right)\left( {\frac{{{V_s} + {V_T}}}{{2\pi {V_T}}}} \right)}
\end{array}
\eqno{(34)}
$$
\end{proof}
An analysis can be performed in order to view the implications of having different ratios between the sensor length and the initial evader region radius on the number of iterations and the times it takes to reduce the evader region to be bounded by a circle with a radius that is equal to or smaller than $r$. We will assume that $R_0$ can be expressed as
$$
R_0 = \alpha r,\alpha  > 1
\eqno{(35)}
$$
In Fig. $6$. we can view the number of iterations and cleaning times of the sweep process until the evader region is bounded by a circle with a radius that is equal to or smaller than $r$, for $1 < \alpha \leq 100$.
\begin{figure}[ht]
\noindent \centering{}\includegraphics[width=3.4in,height =4in]{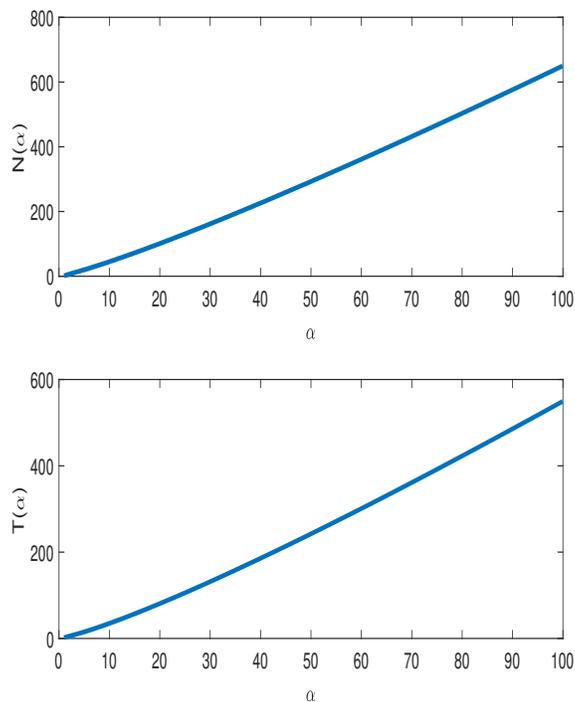} \caption{ Number of iterations and cleaning times for $1\le \alpha \leq 100$ for a single agent or alternatively a line formation of agents employing the circular sweep process. The parameters values chosen for this plot are $V_T = 1$,$\Delta V = V_T$.}
\label{Fig6Label}
\end{figure}
From Fig. $6$. we can view that as $\alpha$ increases $N$ and $T$ increase in a close to linear manner. This is not intuitive when looking at the equations that are derived for $N$ and $T$.
Another interesting analysis can be done in order to view the implications of different choices of $\Delta V$ on the number of iterations and termination times of the cleaning process.
\begin{figure}[ht]
\noindent \centering{}\includegraphics[width=3.4in,height =4in]{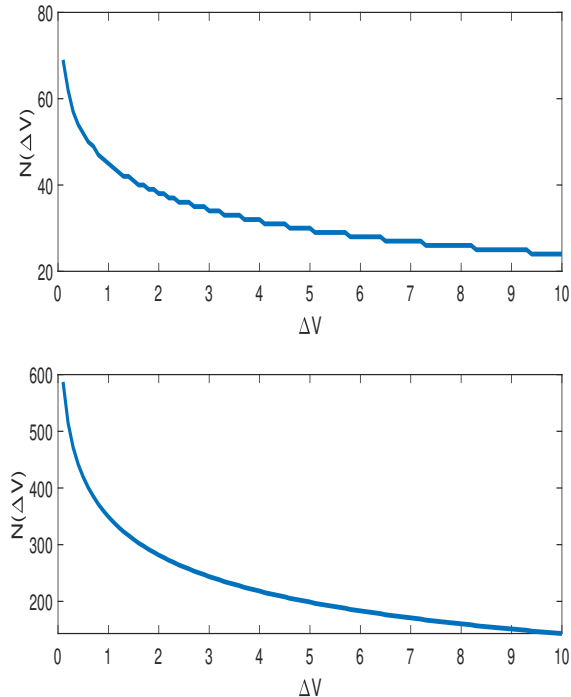} \caption{Number of iterations and cleaning times for different choices of $\Delta V$ for a single agent or alternatively a line formation of agents employing the circular sweep process. In this simulation $\Delta V$ varies between $0.1V_T\leq \Delta V \leq 10V_T$. The other parameters values chosen for this plot are $V_T = 1$, $R_0=100$ and $r=10$.}
\label{Fig7Label}
\end{figure}
From Fig. $7$. We can view that as $\Delta V$ increases $N$ and $T$ decrease in a piecewise exponential manner. This is result can be anticipated since we use the number of iterations has to be an integer number. This results in that for close values of $\Delta V$ the number of iterations will stay the same as for a slightly smaller value of $\Delta V$. Only when $\Delta V$ is sufficiently large in order to cause the iteration number to decrease by a number that is greater than $1$ iteration this result will be apparent in the plot.
\section{The End Game}
In order to entirely clean the evader region the sweeper formation needs to change the scanning method when the evader region is bounded by a circle of radius $r$. This is due to the fact that a smart evader that is very close to the center of the evader region can travel at a very high angular velocity compared to the angular velocity of the searcher. This constraint is described by the following two equations, ${\omega _s} = \frac{{{V_s}}}{r}, {\omega _T} = \frac{{{V_T}}}{\varepsilon }$. The first describes the searcher's angular velocity and the second the evader's angular velocity. Since $\varepsilon$ can be arbitrarily small the evader can move just behind the sweeper's sensor and never be detected. Thus a slight modification to the sweep process needs to be applied in order to clean the entire evader region with the line formation of agents that employs a circular scan. After completing sweep number $N-1$ the sweepers move towards the center of the evader region until the lower tip of the sensor of the closest agent to the center of the evader region is placed at the center of the evader region. Following this motion the sweepers perform a circular sweep of radius $r$ around the center of the evader region. The time this last circular sweep takes is given by ${T_{last}} = \frac{{2\pi r}}{{{V_s}}}$. Following this last scan the sweepers advance a distance of $r$ downwards until the lower part of the formation's sensor is located at the point $(-r,0)$. The time it takes the sweeper to perform this movement is given by, ${T_l} = \frac{r}{V_s}$. Therefore after the last circular scan and the last inward motion the evader region is bounded by a circle of radius ${R_{last}}$, given by,
$$
{R_{last}} = {T_{last}}{V_T} + {T_l}{V_T} = \frac{{r{V_T}\left( {2\pi  + 1} \right)}}{{{V_s}}}
\eqno{(36)}
$$
For $R_0 = 100, r=10, V_T = 1,V_s = 64.8319$, ${R_{last}} =1.1234 $. In order to overcome the challenges in the circular search that were described we propose that after scan number $N+1$ the agent line formation will travel to the right until cleaning the wave front that propagates from the right portion of the remaining evader region and then travel to the left until cleaning the remaining evader region. A depiction of the scenario at the beginning of the end game is presented in Fig. $8$.
\begin{figure}[ht]
\noindent \centering{}\includegraphics[width=2.5in,height =2.5in]{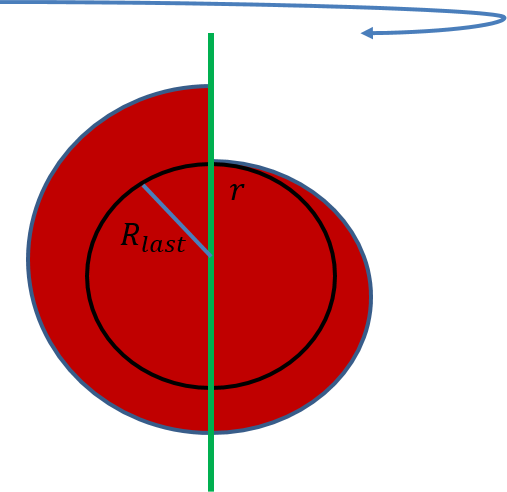} \caption{Depiction of the linear right and left last sweep. In order to overcome the challenges in the circular search that were described we propose that after scan number $N+1$ the agent will travel to the right until cleaning the wavefront that propagates from the right portion of the remaining evader region and then travel to the left until cleaning the remaining evader region.}
\label{Fig8Label}
\end{figure}
Theorem $7$ states the conditions for this demand to hold.
\begin{thm}
When defining $\alpha= \frac{{R_0}}{r}$, if $\Delta V$ satisfies that,
$$
\Delta V \ge  - 2\pi \alpha {V_T} + \pi {V_T} + {V_T} + {V_T}\sqrt {{\pi ^2} + 6\pi  + 7}
\eqno{(37)}
$$
then the evader region will be completely cleaned by a single agent or a line formation of agents that employs the linear scan after $N+1$ iterations.
\end{thm}
\begin{proof}
During the previously mentioned movement the margin between the edge of the sensor in each direction to the evader region boundaries must satisfy,
$$
\frac{{r - {R_{last}}}}{{{V_T}}} > {T_{one}}
\eqno{(38)}
$$
in order to guarantee no escape. ${T_{one}}$ denotes the time it takes the linear formation to clean the right section of the remaining evader region in addition to the time it takes it to scan from the rightmost point it got to until the leftmost point of the expansion. These times are respectively denoted as $t$ and $ \tilde t$. Therefore ${T_{one}}$ is given by ${T_{one}} = \tilde t + t$. The evader region's rightmost point of expansion starts from the point $(R_{last},0)$ and spreads at a velocity of $V_T$. Therefore if the constraint in $(38)$ is satisfied we can view the rightwards and leftwards linear sweeps as a one dimensional scan. This geometric constraint can be observed in Fig. $8$. Therefore the time $t$ it will take the formation to clean the spread of potential evaders from the right section of the region can be calculated from, ${V_s}t = {R_{last}} + {V_T}t$. Therefore $t$ is given by, $t = \frac{{{R_{last}}}}{{{V_s} - {V_T}}}$. $\tilde t$ is computed by calculating the time it will take the linear formation located at point $(tV_s,0)$ to change its scanning direction and perform a leftwards scan to a point that spread at a velocity of $V_T$ from the leftmost point in the evader region at the origin of the search, the point $(-R_{last},0)$, for a time given by ${\tilde t + t}$. We have that, $- {R_{last}} - {V_T}\left( {\tilde t + t} \right) = t{V_s} - {V_s}\tilde t$. Plugging in the value of $t$ yields $\tilde t = \frac{{2{V_s}{R_{last}}}}{{{{\left( {{V_s} - {V_T}} \right)}^2}}}$. $T_{one}$ is therefore given by,
$$
{T_{one}} = t + \tilde t = \frac{{{r^2}{V_T}\left( {2\pi  + 1} \right)\left( {6\pi {R_0}{V_T} + 3\Delta Vr + 2{V_T}r} \right)}}{{{V_s}{{\left( {2\pi {R_0}{V_T} + \Delta Vr} \right)}^2}}}
\eqno{(39)}
$$
For $R_0 = 100, r=10, V_T = 1, V_s =64.8319$, $t=0.0176 ,\tilde t=0.0358$ and ${T_{one}} =0.0533$. Therefore the total scan time until a complete cleaning of the evader region is given by $T_{total} = T + T_{one} =349.3854$. For the one dimensional scan to be valid and insure a non escape search and complete cleaning of the evader region $(38)$ must be satisfied. This demand implies that,
$$
\frac{{r - {R_{last}}}}{{{V_T}}} > \frac{{{R_{last}}\left( {3{V_s} - {V_T}} \right)}}{{{{\left( {{V_s} - {V_T}} \right)}^2}}}
\eqno{(40)}
$$
From substitution of the expressions for $V_s$ and ${R_{last}}$, $(40)$ can be written as,
$$
r{\left( {{V_s} - {V_T}} \right)^2} > {R_{last}}{V_s}\left( {{V_T} + {V_s}} \right)
\eqno{(41)}
$$
By substitution of $R_0$ with $\alpha r$ where $\alpha  > 0$ and by substituting the terms for $V_s$ and ${R_{last}}$, $(41)$ resolves to a quadratic equation in $\Delta V$ that has only one positive root. This root is a monotonically decreasing function in $\alpha$, given by
$$
\Delta V \ge  - 2\pi \alpha {V_T} + \pi {V_T} + {V_T} + {V_T}\sqrt {{\pi ^2} + 6\pi  + 7}
\eqno{(42)}
$$
Therefore for a given $\alpha$, the designer of the sweep process can infer which $\Delta V $ needs to be chosen in order to completely clean the evader region using the final linear sweeping motion.
\end{proof}
\begin{figure}[ht]
\noindent \centering{}\includegraphics[width=4in,height =2.5in]{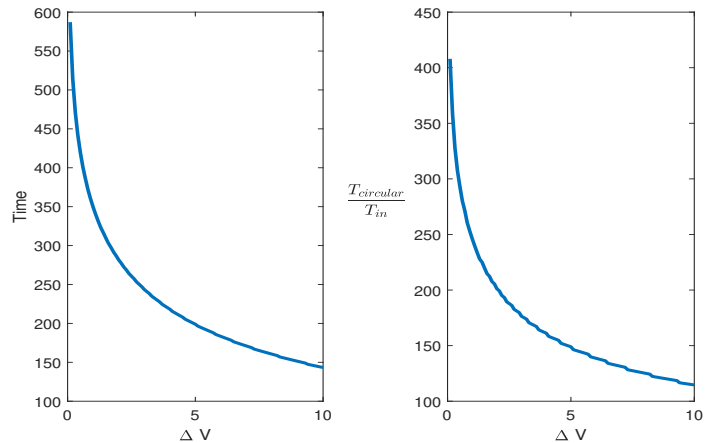} \caption{The left plot shows the complete cleaning times of the evader region with a linear formation that performs the circular sweep process as a function of $\Delta V$. The right plot shows the ratio between the circular and inward advancement times during the sweep process.}
\label{Fig9Label}
\end{figure}
The left plot of Fig. $9$. shows the complete cleaning times of the evader region with a linear formation as a function of $\Delta V$. The right plot emphasizes that most of the cleaning times are due to circular sweeping motions, and that the inward advancement times have only a small effect on the total search time. For a complete derivation see Appendix $K$.
\begin{thm}
For a valid circular search process the total scan time until a complete cleaning of the evader region is given by, $T = {T_{circular}} + {T_{in}} + {T_{one}}$, or as,
$$
\begin{array}{l}
T =  - \frac{{{R_0}\left( {{V_s} + {V_T}} \right)}}{{{V_s}{V_T}}} + \frac{{r\left( {{V_s} - {V_T}} \right)\left( {{V_s} + {V_T} + 2\pi {V_T}N} \right)}}{{2\pi {V_s}{V_T}^2}} + \\
{\left( {1 + \frac{{2\pi {V_T}}}{{{V_s} + {V_T}}}} \right)^{N - 1}}\left( {\frac{{\left( {2\pi {R_0}{V_T} - r\left( {{V_s} - {V_T}} \right)} \right)\left( {{V_s} + {V_T} + 4\pi {V_T}} \right)}}{{2\pi {V_s}{V_T}^2}}} \right) + \frac{{2\pi r}}{{{V_s}}} + \frac{{{R_0}}}{{{V_s}}} + \frac{{{r^2}{V_T}\left( {2\pi  + 1} \right)\left( {6\pi {R_0}{V_T} + 3\Delta Vr + 2{V_T}r} \right)}}{{{V_s}{{\left( {2\pi {R_0}{V_T} + \Delta Vr} \right)}^2}}}
\end{array}
\eqno{(43)}
$$
\end{thm}

\section{CONCLUSIONS}
This research analyses a scenario in which a line formation of agents, or alternatively a single agent, prevent escape from an initial circular area containing evaders that must be detected. We prove interesting and nonintuitive results on how to determine the minimal velocity a sweeper agent should have, in order to reduce the evader region to a circle with a smaller radius than the searcher's sensor length. We prove that the minimal searcher velocity preventing escape from the evader region cannot be solely based upon a single circular traversal around the evader region. We compare this critical velocity to a lower bound derived on the velocity, that is independent of the search process. We then show that for the line formation of agents it is impossible to completely clean the area using only a circular sweeping motion. Therefore after the evader region is shrunk to a circle with a radius of less than $r$ a modification to the search process is introduced. We then show that if the ratio between the searcher velocity $V_s$ and the evader's maximal velocity $V_T$ is above a derived threshold than the sweeper formation completely cleans the region.

\bibliographystyle{unsrt}
\bibliography{References}

\appendices
\section{}
\textbf{Theorem 4. }
The function $f(t,V_s)$ reaches its minimum at time $t^*(V_s)$, where $t^*(V_s)$ is given by,
$$
{t^*(V_s)} = \sqrt {\frac{{ - b - \sqrt {{b^2} - 4ac} }}{{2a}}}  - \frac{{2\pi {R_0}}}{{{V_s}}}
\eqno{(44)}
$$
Where the coefficients $a,b,c$ are given by,
$$
a = {k^2}{V_T}^4,b = l - 2{k^2}{r^2}{V_T}^2 - 2k{V_T}^2,c = 2k{r^2} + {k^2}{r^4}
\eqno{(45)}
$$
With $k,l$ given by,
$$
k = \frac{1}{{2{R_0}({R_0} + r)}},l = \frac{{{V_T}^4}}{{{V_s}^2{{({R_0} + r)}^2}}}
\eqno{(46)}
$$
\begin{proof}
We have that,
$$
f(t,{V_s}) = 1 + \frac{1}{{2{R_0}({R_0} + r)}}\left( {{r^2} - {V_T}^2{{\left( {\frac{{2\pi {R_0}}}{{{V_s}}} + t} \right)}^2}} \right) - \cos \left( {\frac{{{V_s}t}}{{{R_0}}}} \right)
\eqno{(47)}
$$
We denote by $M$ the following expression,
$$
M = {\left( {\frac{{2\pi {R_0}}}{{{V_s}}} + t} \right)^2}
\eqno{(48)}
$$
Taking the derivative of $f(t,{V_s})$ with respect to $t$ yields,
$$
f'(t) =  - \frac{{{V_T}^2}}{{{R_0}({R_0} + r)}}\left( {\frac{{2\pi {R_0}}}{{{V_s}}} + t} \right) + \sin \left( {\frac{{{V_s}t}}{{{R_0}}}} \right)\frac{{{V_s}}}{{{R_0}}}
\eqno{(49)}
$$
Equating $(49)$ to $0$ in order to find the minimum point of $f(t,{V_s})$ yields,
$$
\sin \left( {\frac{{{V_s}t}}{{{R_0}}}} \right) = \frac{{{V_T}^2}}{{{V_s}({R_0} + r)}}\left( {\frac{{2\pi {R_0}}}{{{V_s}}} + t} \right)
\eqno{(50)}
$$
Use of a geometric relation results in,
$$
\cos \left( {\frac{{{V_S}t}}{{{R_0}}}} \right) = \sqrt {1 - {{\sin }^2}\left( {\frac{{{V_s}t}}{{{R_0}}}} \right)}
\eqno{(51)}
$$
Substitution of the expressions in $(48)$,$(50)$  into $(51)$ yields,
$$
\cos \left( {\frac{{{V_S}t}}{{{R_0}}}} \right) = \sqrt {1 - \frac{{{V_T}^4M}}{{{V_s}^2{{({R_0} + r)}^2}}}}
\eqno{(52)}
$$
Throughout the proof we will use the values of ${R_0} = 100$, $r=10$, ${V_T} =1$. We denote,
$$
k = \frac{1}{{2{R_0}({R_0} + r)}} = \frac{1}{{22000}}, l=\frac{{{V_T}^4}}{{{V_s}^2{{({R_0} + r)}^2}}} = \frac{1}{{{{110}^2}{V_s}^2}}
\eqno{(53)}
$$
Therefore $(47)$ takes the form of,
$$
f(t,{V_s}) = 1 + k\left( {{r^2} - {V_T}^2M} \right) - \sqrt {1 - lM}
\eqno{(54)}
$$
Or,
$$
f(t,{V_s}) = 1 + \frac{1}{{22000}}\left( {100 - M} \right) - \sqrt {1 - \frac{1}{{{{110}^2}{V_s}^2}}M}
\eqno{(55)}
$$
We desire to find the minimal $V_s$ that results in $f(t,{V_s})\geq 0 \hspace{1mm} \forall \hspace{1mm} t$. We look for the zero crossing of $f(t,{V_s})$ and therefore equate $(55)$ to $0$. This yields that,
$$
1 - lM = {\left( {1 + k\left( {{r^2} - {V_T}^2M} \right)} \right)^2}
\eqno{(56)}
$$
Or,
$$
1 - \frac{1}{{{{110}^2}{V_s}^2}}M = {\left( {1 + \frac{1}{{22000}}\left( {100 - M} \right)} \right)^2}
\eqno{(57)}
$$
Rearranging terms yields,
$$
1 - lM = 1 + {k^2}{\left( {{r^2} - {V_T}^2M} \right)^2} + 2k\left( {{r^2} - {V_T}^2M} \right)
\eqno{(58)}
$$
Or,
$$
- \frac{1}{{121{V_s}^2}}M = \frac{1}{{4840000}}{\left( {100 - M} \right)^2} + \frac{{\left( {100 - M} \right)}}{{110}}
\eqno{(59)}
$$
Further development of terms yields,
$$
- \frac{{40000M}}{{{V_s}^2}} = 10000 - 200M + {M^2} + 4400000 - 44000M
\eqno{(60)}
$$
Choosing a value of,
$$
{V_s} = \frac{{2\pi {R_0}{V_T}}}{r}
\eqno{(61)}
$$
Yields the two solutions of the quadratic equation for $M$ in $(60)$ given by,
$$
{M_1} = 100.0023, {M_2} = 44089.844
\eqno{(62)}
$$
Only ${M_1}$ is in our domain of interest. Generally the quadratic equation for $M$ takes the form of,
$$
- lM = {k^2}{r^4} + {k^2}{V_T}^4{M^2} - 2{k^2}{r^2}{V_T}^2M + 2k{r^2} - 2k{V_T}^2M
\eqno{(63)}
$$
Rearrangement of terms yields,
$$
{k^2}{V_T}^4{M^2} + M\left( {l - 2{k^2}{r^2}{V_T}^2 - 2k{V_T}^2} \right) + 2k{r^2} + {k^2}{r^4}
\eqno{(64)}
$$
We denote the quadratic equation's coefficients by
$$
a = {k^2}{V_T}^4,b = l - 2{k^2}{r^2}{V_T}^2 - 2k{V_T}^2,c = 2k{r^2} + {k^2}{r^4}
\eqno{(65)}
$$
Only the smallest of the equation's roots lies in our domain on interest that correspond to times of a up to a quarter of a traversal around the evader region. Therefore we have that,
$$
{\left( {\frac{{2\pi {R_0}}}{{{V_S}}} + {t^*}} \right)^2} = \frac{{ - b - \sqrt {{b^2} - 4ac} }}{{2a}}
\eqno{(66)}
$$
And the solution for ${t^*(V_s)}$ is given by,
$$
{t^*(V_s)} = \sqrt {\frac{{ - b - \sqrt {{b^2} - 4ac} }}{{2a}}}  - \frac{{2\pi {R_0}}}{{{V_s}}}
\eqno{(67)}
$$
Substituting coefficients for the chosen values yields,
$$
{t^*(V_s)} = 0.0012
\eqno{(68)}
$$
$\openbox$.
\end{proof}
Figure $10$ shows that $f'(t,V_s)$ is an increasing function.
\begin{figure}[ht]
\noindent \centering{}\includegraphics[width=3.5in,height =3in]{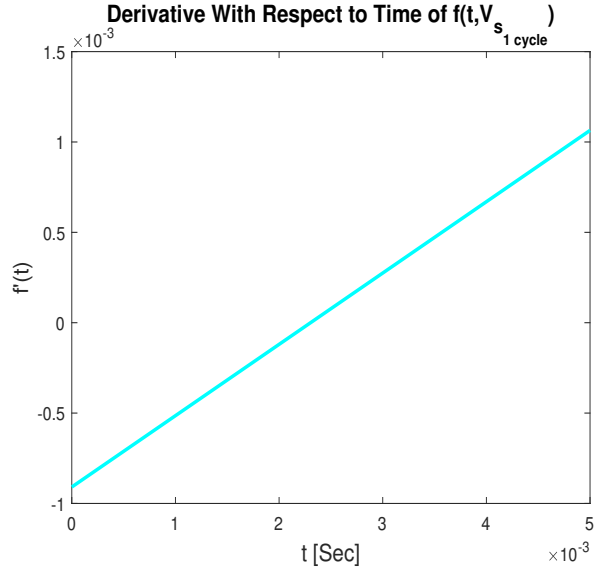} \caption{plot of the derivative of $f(t)$. It can be seen that very close to the start time of the simulation $f'(t)$ is negative hence $f(t)$ decreases in this area. The parameters values chosen for this plot are $r=10$, $V_T = 1$ and $R_0 = 100$.}
\label{Fig10Label}
\end{figure}

Fig. $11$. shows that basing the critical velocity on the naive proposal for a circular search pattern does not lead to satisfaction of the confinement task. Plugging the values of $r=10$, $V_T = 1$ and $R_0 = 100$ results in $f(t^*,V_{s_{\hspace{1mm} 1 \hspace{1mm} cycle}}) =  - 1.047*{10^{ - 6}}<0$ as can be observed in Fig. $11$. This validates our proof that there exists a set of search parameters for which $V_{s_{\hspace{1mm} 1 \hspace{1mm} cycle}})$ is not sufficient.
\begin{figure}[ht]
\noindent \centering{}\includegraphics[width=3.2in,height =2.8in]{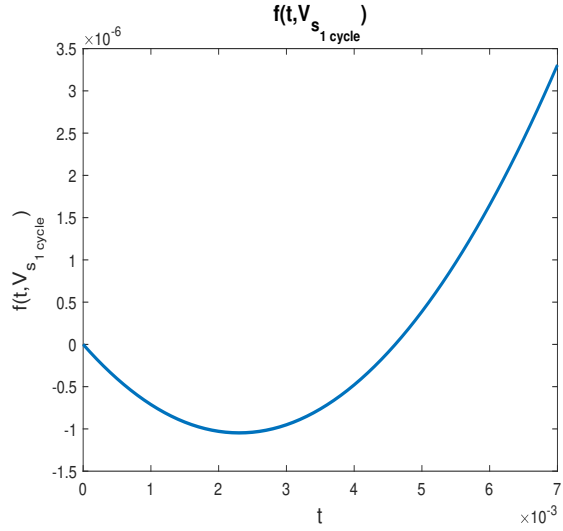} \caption{It can be seen that when basing the critical velocity only upon one cycle the desired inequality of $f(t,V_{s_{\hspace{1mm} 1 \hspace{1mm} cycle}}) \geq 0 $ is not satisfied for all desired times. The minimum point of $f(t,V_{s_{\hspace{1mm} 1 \hspace{1mm} cycle}}$ occurs at $f(t^*,V_{s_{\hspace{1mm} 1 \hspace{1mm} cycle}})$. The parameters values chosen for this plot are $r=10$, $V_T = 1$ and $R_0 = 100$.}
\label{Fig11Label}
\end{figure}
Fig. $12$. shows the dependence of $t^*$ on $V_s$.
\begin{figure}[ht]
\noindent \centering{}\includegraphics[width=3.2in,height =2.8in]{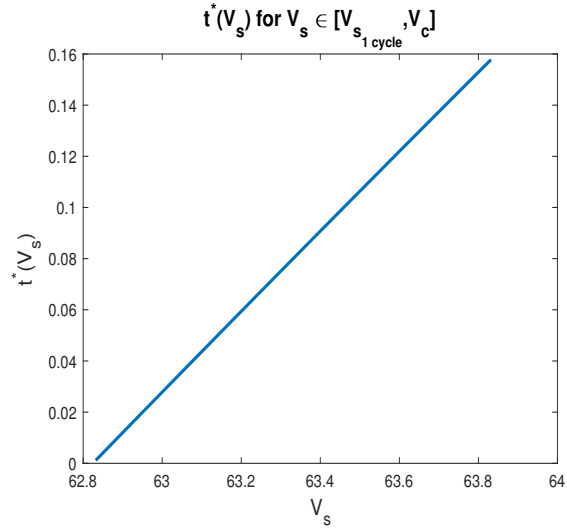} \caption{Dependence of $t^*$ on $V_s$. It can be seen that $t^*(V_s)$ is a monotonically increasing function in $V_s$ and therefore for a larger $V_s$ that ensures that $\forall \hspace{1 mm} t ,\hspace{1 mm} f(t,V_s) \geq 0$, $t^*$ will occur later in the search process. The parameters values chosen for this plot are $r=10$, $V_T = 1$ and $R_0 = 100$.}
\label{Fig12Label}
\end{figure}
\section{}
\textbf{Theorem 5. }
For all search parameters satisfying that $R_0\geq r$ it holds that  $f(t^*,V_{c}) > 0$.  Therefore it holds that $\forall \hspace{1 mm} t ,\hspace{1 mm} f(t^*,V_{c}) \geq 0$. Thus $V_{c}$ is a sufficiently high velocity in order to accomplish the confinement task.
\begin{proof}
We first assume  that $R_0$ can be expressed as $R_0 = \alpha r$ , where $\alpha  \geq 1$. Without loss of generality we assume that $r=1$ and $V_T =1$. We can assume these values since all the other search parameters depend on them. In section $6$ we saw that $V_c$ can be expressed as,
$$
{V_c} = \frac{{2\pi {R_0}{V_T} + {V_T}r}}{r}
\eqno{(69)}
$$
Substituting $R_0 = \alpha r$, $r=1$ and $V_T =1$ into $(69)$ yields,
$$
{V_c} = 2\pi \alpha  + 1
\eqno{(70)}
$$
$V_{c}$ is a monotonically increasing function in $\alpha$.
$$
f({t^*},{V_C}) = 1 + \frac{1}{{2\alpha (\alpha  + 1)}}\left( {1 - {{\left( {\frac{{2\pi \alpha }}{{2\pi \alpha  + 1}} + {t^*}} \right)}^2}} \right) - \cos \left( {\frac{{{V_C}{t^*}}}{{{R_0}}}} \right)
\eqno{(71)}
$$
$f(t,V_c)$ is monotonically decreasing function in $\alpha$. The argument inside the cosine function decays to $0$ as alpha increases since $t^*$ goes to as alpha increases. The second argument in $(71)$ also tends to $0$ as $\alpha$ increases. This behaviour in plotted in Fig. $13$. Therefore in order to lower bound $f(t^*,V_c)$ we will prove that,
$$
\mathop {\lim }\limits_{\alpha  \to \infty } f({t^*},{V_c}) \ge 0
\eqno{(72)}
$$
since $t^*$ is the point in time in which $f(t,V_c)$ is minimal, ensuring that $f({t^*},{V_c}) \ge 0$ guarantees that $f({t},{V_c}) \ge 0 \forall \hspace{1 mm} t $. Following the notation in the proof of $t^*$ we have that,
$$
k = \frac{1}{{2\alpha (\alpha  + 1)}}
\eqno{(73)}
$$

$$
l = \frac{{{V_T}^4}}{{{V_c}^2{{({R_0} + r)}^2}}} = \frac{1}{{{{\left( {2\pi \alpha  + 1} \right)}^2}{{(\alpha  + 1)}^2}}}
\eqno{(74)}
$$
The quadratic equation coefficients for the solution of $M$ are given by,
$$
a = {k^2}{V_T}^4 = \frac{1}{{4{\alpha ^2}{{(\alpha  + 1)}^2}}}
\eqno{(75)}
$$

$$
b = l - 2{k^2}{r^2}{V_T}^2 - 2k{V_T}^2 = \frac{1}{{{{\left( {2\pi \alpha  + 1} \right)}^2}{{(\alpha  + 1)}^2}}} - \frac{1}{{2{\alpha ^2}{{(\alpha  + 1)}^2}}} - \frac{1}{{\alpha (\alpha  + 1)}}
\eqno{(76)}
$$

$$
c = 2k{r^2} + {k^2}{r^4} = \frac{1}{{\alpha (\alpha  + 1)}} + \frac{1}{{4{\alpha ^2}{{(\alpha  + 1)}^2}}} = \frac{{4\alpha (\alpha  + 1) + 1}}{{4{\alpha ^2}{{(\alpha  + 1)}^2}}}
\eqno{(77)}
$$

Letting $\alpha \to \infty$ we have,
$$
\mathop {\lim }\limits_{\alpha  \to \infty } a = \frac{1}{{4{\alpha ^4}}}
\eqno{(78)}
$$

$$
\mathop {\lim }\limits_{\alpha  \to \infty } b =  - \frac{{8{\pi ^2}{\alpha ^4}}}{{8{\pi ^2}{\alpha ^6}}} =  - \frac{1}{{{\alpha ^2}}}
\eqno{(79)}
$$

$$
\mathop {\lim }\limits_{\alpha  \to \infty } c = \frac{1}{{{\alpha ^2}}}
\eqno{(80)}
$$
Selecting the root that is the interval of interest yields,
$$
{\left( {\frac{{2\pi {R_0}}}{{{V_c}}} + {t^*}} \right)^2} = \frac{{ - b - \sqrt {{b^2} - 4ac} }}{{2a}}
\eqno{(81)}
$$
Letting $\alpha \to \infty$ we have,
$$
\mathop {\lim }\limits_{\alpha  \to \infty } \frac{{ - b - \sqrt {{b^2} - 4ac} }}{{2a}} \approx 2{\alpha ^4}\left( {\frac{1}{{{\alpha ^2}}} - \frac{1}{{{\alpha ^2}}}\sqrt {1 - \frac{1}{{{\alpha ^2}}}} } \right)
\eqno{(82)}
$$
The Taylor Series expression for $\sqrt {1 - \frac{1}{{{\alpha ^2}}}}$ yields,
$$
\sqrt {1 - \frac{1}{{{\alpha ^2}}}}  \approx 1 - \frac{1}{{2{\alpha ^2}}} + o\left( {\frac{1}{{{\alpha ^2}}}} \right)
\eqno{(83)}
$$
We therefore obtain that $(82)$ satisfies,
$$
\mathop {\lim }\limits_{\alpha  \to \infty }  \approx \mathop {\lim }\limits_{\alpha  \to \infty } 2{\alpha ^4}\left( {\frac{1}{{{\alpha ^2}}} - \frac{1}{{{\alpha ^2}}}\left( {1 - \frac{1}{{2{\alpha ^2}}} + o\left( {\frac{1}{{{\alpha ^2}}}} \right)} \right)} \right) \to 1
\eqno{(84)}
$$
$t^*(V_c)$ is therefore given by,
$$
{t^*(V_c)} = \sqrt {\frac{{ - b - \sqrt {{b^2} - 4ac} }}{{2a}}}  - \frac{{2\pi {R_0}}}{{{V_c}}}
\eqno{(85)}
$$
Letting $\alpha \to \infty$ yields,
$$
\mathop {\lim }\limits_{\alpha  \to \infty } {t^*(V_c)} = \mathop {\lim }\limits_{\alpha  \to \infty } 1 - \frac{{2\pi \alpha }}{{2\pi \alpha  + 1}} \to 0
\eqno{(86)}
$$
Taking the limit of $\alpha \to \infty$ on the argument inside the cosine function in $f({t^*},{V_c})$ yields,
$$
\mathop {\lim }\limits_{\alpha  \to \infty } \frac{{{V_c}{t^*}}}{{{R_0}}} = \mathop {\lim }\limits_{\alpha  \to \infty } \frac{{2\pi \alpha  + 1}}{\alpha }{t^*} \to 0
\eqno{(87)}
$$
Therefore we obtain that,
$$
\mathop {\lim }\limits_{\alpha  \to \infty } f({t^*},{V_c}) = 0
\eqno{(88)}
$$
This results in the desired behaviour that guarantees that,
$$
f(t,{V_c}) \ge 0,\forall \hspace{1mm}  t, \alpha \geq 1
\eqno{(89)}
$$
\end{proof}
$\openbox$.

\begin{figure}[ht]
\noindent \centering{}\includegraphics[width=3in,height =2.5in]{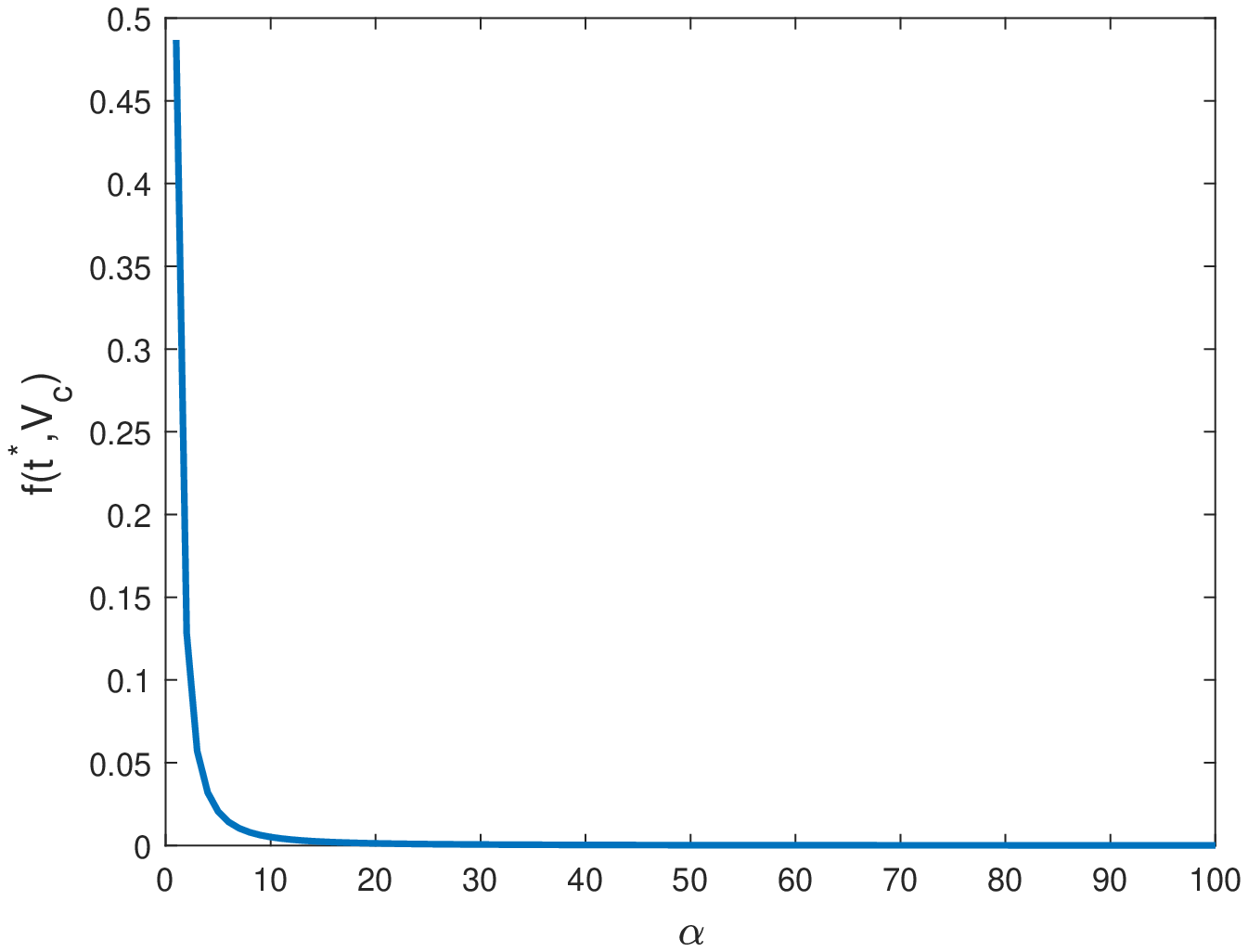} \caption{$f(t^*(V_c),V_c)$ as a function of $\alpha$. For this plot $V_T =1$, $r=1$.}
\label{Fig13Label}
\end{figure}

\begin{figure}[ht]
\noindent \centering{}\includegraphics[width=3in,height =2.5in]{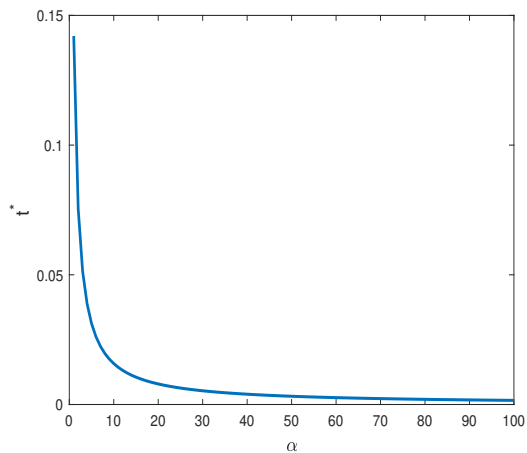} \caption{$t^*(V_c)$ as a function of $\alpha$. For this plot $V_T =1$, $r=1$.}
\label{Fig14Label}
\end{figure}
\section{}
In section $5$ we proved that for a sweeper velocity of $V_{s_{\hspace{1mm} 1 \hspace{1mm} cycle}}$, $f(t^*,V_{s_{\hspace{1mm} 1 \hspace{1mm} cycle}})$ is slightly less than $0$ for all choices of parameters. In Theorem $5$ we proved that for a choice of ${V_C} = \frac{{2\pi {R_0}{V_T} + {V_T}r}}{r}$, $f(t^*,V_C)$ is slightly greater than $0$ for all choices of parameters. We can therefore apply the bisection algorithm between the two velocities. We pick a desired tolerance parameter $\varepsilon$. Clearly $V_{s_{\hspace{1mm} 1 \hspace{1mm} cycle}}<{V_C}$. We proved that $f(t^*,V_{s_{\hspace{1mm} 1 \hspace{1mm} cycle}})<0$ and that $f(t^*,V_C)>0$ therefore $f(t^*,V_{s_{\hspace{1mm} 1 \hspace{1mm} cycle}})f(t^*,V_C) < 0$. \newline We initialize ${l_0} = f(t^*,V_{s_{\hspace{1mm} 1 \hspace{1mm} cycle}})$, ${u_0} = {V_C}$. \newline For any $k=0,1,2,...$ we execute the following steps until the desired accuracy is obtained:\newline
$(1)$ We choose ${x_k} = \frac{{{u_k} + {l_k}}}{2}$. \newline
$(2)$ If $f\left( {{l_k}} \right)f\left( {{x_k}} \right) > 0$, we define ${l_{k + 1}} = {x_k}$, ${u_{k + 1}} = {u_k}$. Otherwise we define ${l_{k + 1}} = {x_k}$, ${u_{k + 1}} = {x_k}$. \newline
$(3)$ If ${u_{k + 1}} - {l_{k + 1}} \le \varepsilon$ then stop, and $x_k$ is the output.
\\
Fig. $15$. shows a plot of $f(t,V_{s_{\hspace{1mm} bisection}})$. $V_{s_{\hspace{1mm} bisection}}$ obtained from the bisection method. It can now be observed that for times of interest $\left| {f(t^*,{V_{{s_{bisection}}}})} \right| \le \varepsilon$ therefore  $f(t,{V_{{s_{bisection}}}}) + \varepsilon  \ge 0 \forall \hspace{1mm} t $ and hence ensures a guaranteed no escape searcher velocity for all possible evader trajectories satisfying a maximal evader velocity of $V_T$, with the smallest possible $V_s$.
\begin{figure}[ht]
\noindent \centering{}\includegraphics[width=3.4in,height =3in]{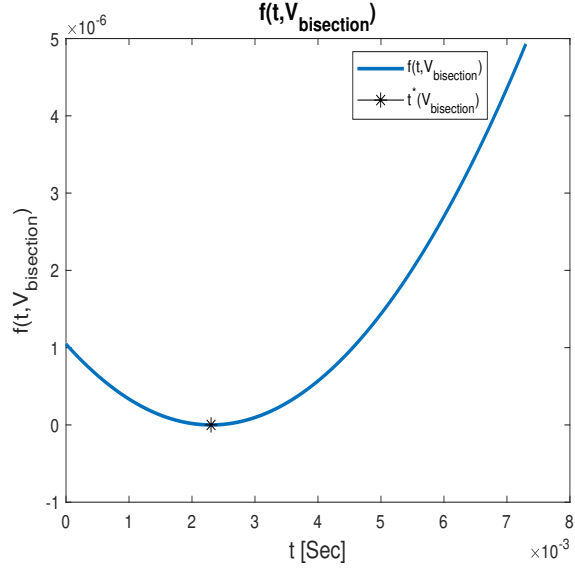} \caption{$f(t,V_s)$ values for $V_s = V_{s_{\hspace{1mm} bisection}}$. The black asterisk denotes the appropriate $t^*(V_{s_{\hspace{1mm} bisection}})$ of $f(t,V_{s_{\hspace{1mm} bisection}})$. The parameters values chosen for this plot are $r=10$, $V_T = 1$ and $R_0 = 100$.}
\label{Fig15Label}
\end{figure}
\section{}
The number of sweep iterations that are required to reduce the evader region to be bounded by a circle with a radius that is less or equal to $R_N$ is calculated in the following manner. We have that,
$$
{R_{i + 1}} = {c_3}{R_i} + {c_1}
\eqno{(90)}
$$
Therefore,
$$
{R_N} = {c_3}^N{R_0} + {c_1}\sum\limits_{i = 0}^{N - 1} {{c_3}^i}  = {c_3}^N\left( {{R_0} - \frac{{{c_1}}}{{1 - {c_3}}}} \right) + \frac{{{c_1}}}{{1 - {c_3}}}
\eqno{(91)}
$$
Rearranging terms results in,
$$
\frac{{{R_N} - \frac{{{c_1}}}{{1 - {c_3}}}}}{{{R_0} - \frac{{{c_1}}}{{1 - {c_3}}}}} = {c_3}^N
\eqno{(92)}
$$
Applying the natural logarithm function to both sides results leads to,
$$
\ln \left( {\frac{{{R_N} - \frac{{{c_1}}}{{1 - {c_3}}}}}{{{R_0} - \frac{{{c_1}}}{{1 - {c_3}}}}}} \right) = N\ln {c_3}
\eqno{(93)}
$$
And the general form for the number of iterations it takes the sweeper swarm to reduce the evader region to be bounded by a circle of a radius that corresponds to the last sweep before completely cleaning the evader region is given by,
$$
N = \left\lceil {\frac{{\ln \left( {\frac{{{R_N} - \frac{{{c_1}}}{{1 - {c_3}}}}}{{{R_0} - \frac{{{c_1}}}{{1 - {c_3}}}}}} \right)}}{{\ln {c_3}}}} \right\rceil
\eqno{(94)}
$$

\section{}
$$
N = \left\lceil {\frac{1}{{\ln {c_3}}}\ln \left( {\frac{{\widehat{{R_N}} - \frac{{{c_1}}}{{1 - {c_3}}}}}{{{R_0} - \frac{{{c_1}}}{{1 - {c_3}}}}}} \right)} \right\rceil
\eqno{(95)}
$$

$$
{c_1} =  - \frac{{r\left( {{V_s} - {V_T}} \right)}}{{{V_s} + {V_T}}},{c_3} = 1 + \frac{{2\pi {V_T}}}{{{V_s} + {V_T}}}
\eqno{(96)}
$$

$$
\frac{{{c_1}}}{{1 - {c_3}}} = \frac{{r\left( {{V_s} - {V_T}} \right)}}{{2\pi {V_T}}}
\eqno{(97)}
$$
Where $\widehat{{R_N}} = r$. We therefore obtain that the number of iterations it takes the line formation of sweepers to reduce the evader region to be bounded by a circle of radius that is smaller than or equal to $r$ is given by,
$$
N = \left\lceil {\frac{{\ln \left( {\frac{{2\pi r{V_T} - r\left( {{V_s} - {V_T}} \right)}}{{2\pi {R_0}{V_T} - r\left( {{V_s} - {V_T}} \right)}}} \right)}}{{\ln \left( {1 + \frac{{2\pi {V_T}}}{{{V_s} + {V_T}}}} \right)}}} \right\rceil
\eqno{(98)}
$$
\section{}
The number of sweep iterations that are required to reduce the evader region to be bounded by a circle with a radius that is less or equal to ${R_{N-2}}$ is calculated in the following manner. We have that,
$$
{R_{i + 1}} = {c_3}{R_i} + {c_1}
\eqno{(99)}
$$
Therefore,
$$
{R_{N - 2}} = {c_3}^{N - 2}{R_0} + {c_1}\sum\limits_{i = 0}^{N - 3} {{c_3}^i}  = {c_3}^{N - 2}\left( {{R_0} - \frac{{{c_1}}}{{1 - {c_3}}}} \right) + \frac{{{c_1}}}{{1 - {c_3}}}
\eqno{(100)}
$$
Rearranging terms results in,
$$
\frac{{{R_{N - 2}} - \frac{{{c_1}}}{{1 - {c_3}}}}}{{{R_0} - \frac{{{c_1}}}{{1 - {c_3}}}}} = {c_3}^{N - 2}
\eqno{(101)}
$$
Therefore the number of sweep iterations that are required to reduce the evader region to be bounded by a circle with a radius that is less or equal to ${R_{N-2}}$ is given by,
$$
{R_{N - 2}} = \frac{{{c_1}}}{{1 - {c_3}}} + {c_3}^{N - 2}\left( {{R_0} - \frac{{{c_1}}}{{1 - {c_3}}}} \right)
\eqno{(102)}
$$
\section{}
The time it takes a multi-agent swarm that performs either the circular sweep or the spiral sweep to completely clean the evader region is calculated as follows. The recursive relation between the next and current radius of the circle that bounds the evader region is given by,
$$
{R_{i + 1}} = {c_3}{R_i} + {c_1}
\eqno{(103)}
$$
Suppose that there exists a constant $\gamma$ such that
$$
\gamma {R_i} = {T_i}
\eqno{(104)}
$$
Therefore multiplying $(103)$ by $\gamma$ on both sides of the equation yields,
$$
{T_{i + 1}} = {c_3}{T_i} + {c_4}
\eqno{(105)}
$$
Where $c_4$ is given by,
$$
{{c_4} = \gamma {c_1}}
\eqno{(106)}
$$
The time it takes to complete the first cycle around the evader region is ${T_0} = \gamma R_0$ and the time it takes to complete the last cycle before the evader region is bounded by a circle of radius $\frac{2r}{n}$ is, that is the time when the evader region is bounded by a circle of with a greater radius than ${R_{N - 1}}$ is given by ${T_{N - 1}} = \gamma {R_{N - 1}}$. Summing over the times of all cycles except the initial one is calculated by summing the cycle times given in $(105)$. This results in,
$$
\sum\limits_{i = 1}^{N - 1} {{T_i}}  = {c_3}\sum\limits_{i = 1}^{N - 1} {{T_i}}  + {c_3}({T_0} - {T_{N - 1}}) + \left( {N - 1} \right){c_4}
\eqno{(107)}
$$
Rearranging terms results in,
$$
\sum\limits_{i = 1}^{N - 1} {{T_i}}  = \frac{{{c_3}({T_0} - {T_{N - 1}}) + \left( {N - 1} \right){c_4}}}{{1 - {c_3}}}
\eqno{(108)}
$$
Since the total time it takes the sweeper swarm to clean the evader region includes also the time of the first sweep we need to add $T_0$ to the summation as well. Thus the total time it takes it takes the sweeper swarm to reduce the evader region to be bounded by a circle of radius that less than or equal to $\frac{2r}{n}$ is given by,
$$
T = \sum\limits_{i = 0}^{N - 1} {{T_i}} = \frac{{{T_0} - {c_3}{T_{N - 1}} + \left( {N - 1} \right){c_4}}}{{1 - {c_3}}}
\eqno{(109)}
$$
\section{}
The time is takes to complete a sweep around the evader region that is bounded by a circle with a radius of ${R_{N-1}}$ is calculated in a similar manner to the calculation in Appendix $F$ . We have that the recursive relation between the time it takes the sweeping agent to complete sweep number $i$ and the time it takes it to complete sweep number $i+1$ is given by,
$$
{T_{i + 1}} = {c_3}{T_i} + {c_4}
\eqno{(110)}
$$
Therefore,
$$
{T_{N - 1}} = {c_3}^{N - 1}{T_0} + {c_4}\sum\limits_{i = 0}^{N - 2} {{c_3}^i}  = {c_3}^{N - 1}\left( {{T_0} - \frac{{{c_4}}}{{1 - {c_3}}}} \right) + \frac{{{c_4}}}{{1 - {c_3}}}
\eqno{(111)}
$$
Rearranging terms results in,
$$
\frac{{{T_{N - 1}} - \frac{{{c_4}}}{{1 - {c_3}}}}}{{{T_0} - \frac{{{c_4}}}{{1 - {c_3}}}}} = {c_3}^{N - 1}
\eqno{(112)}
$$
Therefore the time is takes to complete a sweep around the evader region that is bounded by a circle with a radius of ${R_{N-1}}$ is given by,
$$
{T_{N - 1}} = \frac{{{c_4}}}{{1 - {c_3}}} + {c_3}^{N - 1}\left( {{T_0} - \frac{{{c_4}}}{{1 - {c_3}}}} \right)
\eqno{(113)}
$$
\section{}
The recursive relation between the next and current radius of the circle that bounds the evader region is given by,
$$
{R_{i + 1}} = {c_3}{R_i} + {c_1}
\eqno{(114)}
$$
Summing over the evader region radii up to the $N - 2$ cycle except the initial radius of the evader region is calculated by summing the radii given in $(114)$. This results in,
$$
\sum\limits_{i = 1}^{N - 2} {{R_i}}  = {c_3}\sum\limits_{i = 1}^{N - 2} {{R_i}}  + {c_3}({R_0} - {R_{N - 2}}) + \left( {N - 2} \right){c_1}
\eqno{(115)}
$$
Rearranging terms results in,
$$
\sum\limits_{i = 1}^{N - 2} {{R_i}}  = \frac{{{c_3}({R_0} - {R_{N - 2}}) + \left( {N - 2} \right){c_1}}}{{1 - {c_3}}}
\eqno{(116)}
$$
Since the sum of radii in $(116)$ does include the initial radius of the evader region we need to add $R_0$ to the summation as well. Thus the desired sum of radii is given by,
$$
\sum\limits_{i = 0}^{N - 2} {{R_i}}  = \frac{{{R_0} - {c_3}{R_{N - 2}} + (N - 2){c_1}}}{{1 - {c_3}}}
\eqno{(117)}
$$
\section{}
In this appendix the full derivation of the inward advancement times until the evader region is bounded by a circle with a radius that is less than or equal to $r$ that are denoted by , $\widetilde{{T_{in}}}$, for the case where the sweeper agent employs the circular sweep process are provided. $\widetilde{{T_{in}}}$ is computed by,
$$
\widetilde{{T_{in}}} = \sum\limits_{i = 0}^{{N} - 2} {{T_{i{n_i}}}}= \frac{{\left( {{N} - 1} \right)r}}{{{V_s} + {V_T}}} - \frac{{2\pi {V_T}\sum\limits_{i = 0}^{{N} - 2} {{R_i}} }}{{{V_s}\left( {{V_s} + {V_T}} \right)}}
\eqno{(118)}
$$
The term ${\sum\limits_{i = 0}^{N - 2} {{R_i}} }$ is computed in Appendix $I$ . It is given by,
$$
\sum\limits_{i = 0}^{N - 2} {{R_i}}  = \frac{{{R_0} - {c_3}{R_{{N} - 2}} + (N - 2){c_1}}}{{1 - {c_3}}}
\eqno{(119)}
$$
${R_{N - 2}}$ is calculated in Appendix $F$. It is given by,
$$
{R_{N - 2}} = \frac{{{c_1}}}{{1 - {c_3}}} + {c_3}^{N - 2}\left( {{R_0} - \frac{{{c_1}}}{{1 - {c_3}}}} \right)
\eqno{(120)}
$$
Substituting the coefficients in $(120)$ yields,
$$
{R_{N - 2}} = \frac{{r{V_s}}}{{2\pi {V_T}}} + {\left( {1 + \frac{{2\pi {V_T}}}{{{V_s} + {V_T}}}} \right)^{N - 2}}\left( {\frac{{2\pi {R_0}{V_T} - r{V_s}}}{{2\pi {V_T}}}} \right)
\eqno{(121)}
$$
Substituting the coefficients in $(119)$ yields,
$$
\sum\limits_{i = 0}^{N - 2} {{R_i}}  = \frac{{ - {R_0}\left( {{V_s} + {V_T}} \right)}}{{2\pi {V_T}}} + \frac{{r{V_s}\left( {{V_s} + {V_T} + 2\pi {V_T}} \right)}}{{4{\pi ^2}{V_T}^2}}
 + {\left( {1 + \frac{{2\pi {V_T}}}{{{V_s} + {V_T}}}} \right)^{N - 1}}\left( {\frac{{2\pi {R_0}{V_T} - r{V_s}}}{{4{\pi ^2}{V_T}^2}}} \right)\left( {{V_s} + {V_T}} \right)
 + \frac{{(N - 2)r{V_s}}}{{2\pi {V_T}}}
\eqno{(122)}
$$
Plugging the expression for ${\sum\limits_{i = 0}^{{N} - 2} {{R_i}} }$ from $(122)$ into $(118)$ results in,
$$
{T_{in}} = \sum\limits_{i = 0}^{N - 2} {{T_{i{n_i}}} = } \frac{{\left( {N - 1} \right)r}}{{{V_s} + {V_T}}} - \frac{{2\pi {V_T}}}{{{V_s}\left( {{V_s} + {V_T}} \right)}}\left( {\frac{{{R_0} - {c_3}{R_{N - 2}} + (N - 2){c_1}}}{{1 - {c_3}}}} \right)
\eqno{(123)}
$$
And substituting the developed coefficients into $(123)$ yields,
$$
\widetilde{{T_{in}}} = \sum\limits_{i = 0}^{N - 2} {{T_{i{n_i}}} = } \frac{{{R_0}}}{{{V_s}}} - \frac{r}{{2\pi {V_T}}}\\
 - {\left( {1 + \frac{{2\pi {V_T}}}{{{V_s} + {V_T}}}} \right)^{N - 1}}\left( {\frac{{2\pi {R_0}{V_T} - r{V_s}}}{{2\pi {V_T}{V_s}}}} \right)
\eqno{(124)}
$$
\section{}
$$
{V_s} + {V_T} = \frac{{2\pi {R_0}{V_T} + 2{V_T}r + \Delta Vr}}{r}
\eqno{(125)}
$$

$$
{V_s} - {V_T} = \frac{{2\pi {R_0}{V_T} + \Delta Vr}}{r}
\eqno{(126)}
$$

$$
{\left( {2\pi {R_0}{V_T} + \Delta Vr} \right)^2} > {V_T}r\left( {2\pi  + 1} \right)\left( {2\pi {R_0}{V_T} + 2{V_T}r + \Delta Vr} \right)
\eqno{(127)}
$$
We denote $\alpha= \frac{{R_0}}{r}$ and obtain a quadratic equation in $\Delta V$.
$$
\Delta {V^2} + \Delta V\left( {4\pi \alpha {V_T} - 2\pi {V_T} - 2{V_T}} \right) + 4{\pi ^2}{\alpha ^2}{V_T}^2 - 4{\pi ^2}{V_T}^2\alpha  - 4\pi {V_T}^2 - 4\pi \alpha {V_T}^2 - 6{V_T}^2 > 0
\eqno{(128)}
$$
The roots of $(128)$ a positive and negative root. Since $\Delta V$ is non-negative we are interested only in the positive root. Therefore in order to completely clean the evader region $\Delta V$ has to satisfy
$$
\Delta V \ge  - 2\pi \alpha {V_T} + \pi {V_T} + {V_T} + {V_T}\sqrt {{\pi ^2} + 6\pi  + 7}
\eqno{(129)}
$$
\section{}
The inwards advancement time at iteration $i$ is denoted by ${T_{i{n_i}}}$. It is given by,
$$
{T_{i{n_i}}} = \frac{{{\delta _{{i_{eff}}}}(\Delta V)}}{{{V_s}}} = \frac{{r\left( {{V_s} - {V_T}} \right) - 2\pi {R_i}{V_T}}}{{{V_s}\left( {{V_s} + {V_T}} \right)}}
\eqno{(130)}
$$
The total advancement time until the evader region is bounded by a circle of with a radius that is less than or equal to r is denoted as $\widetilde{{T_{in}}}$. It is given by,
$$
\widetilde{{T_{in}}} = \sum\limits_{i = 0}^{N - 2} {{T_{i{n_i}}}}  = \frac{{\left( {N - 1} \right)r\left( {{V_s} - {V_T}} \right)}}{{{V_s}\left( {{V_s} + {V_T}} \right)}} - \frac{{2\pi {V_T}\sum\limits_{i = 0}^{N - 2} {{R_i}} }}{{{V_s}\left( {{V_s} + {V_T}} \right)}}
\eqno{(131)}
$$
We have that,
$$
{R_{N - 2}} = \frac{{{c_1}}}{{1 - {c_3}}} + {c_3}^{N - 2}\left( {{R_0} - \frac{{{c_1}}}{{1 - {c_3}}}} \right)
\eqno{(132)}
$$
The sum of the radii is given by,
$$
\sum\limits_{i = 0}^{N - 2} {{R_i}}  = \frac{{{R_0} - {c_3}{R_{N - 2}} + (N - 2){c_1}}}{{1 - {c_3}}}
\eqno{(133)}
$$
Where the coefficients $c_1$ and $c_3$ are given by,
$$
{c_1} =  - \frac{{r\left( {{V_s} - {V_T}} \right)}}{{{V_s} + {V_T}}},{c_3} = 1 + \frac{{2\pi {V_T}}}{{{V_s} + {V_T}}}
\eqno{(134)}
$$
Substitution of terms for the expression of ${R_{N - 2}}$ in $(132)$ yields,
$$
{R_{N - 2}} = \frac{{r\left( {{V_s} - {V_T}} \right)}}{{2\pi {V_T}}} + {\left( {1 + \frac{{2\pi {V_T}}}{{{V_s} + {V_T}}}} \right)^{N - 2}}\left( {\frac{{2\pi {R_0}{V_T} - r\left( {{V_s} - {V_T}} \right)}}{{2\pi {V_T}}}} \right)
\eqno{(135)}
$$
Substitution of terms in $(133)$ yields,
$$
\begin{array}{l}
\sum\limits_{i = 0}^{N - 2} {{R_i}}  =  - {R_0}\frac{{{V_s} + {V_T}}}{{2\pi {V_T}}} + \frac{{r\left( {{V_s} - {V_T}} \right)\left( {{V_s} + {V_T} + 2\pi {V_T}} \right)}}{{{{\left( {2\pi {V_T}} \right)}^2}}}\\
 + {\left( {1 + \frac{{2\pi {V_T}}}{{{V_s} + {V_T}}}} \right)^{N - 1}}\left( {\frac{{2\pi {R_0}{V_T} - r\left( {{V_s} - {V_T}} \right)}}{{{{\left( {2\pi {V_T}} \right)}^2}}}} \right)\left( {{V_s} + {V_T}} \right) + \frac{{(N - 2)r\left( {{V_s} - {V_T}} \right)}}{{2\pi {V_T}}}
\end{array}
\eqno{(136)}
$$
We therefore obtain that,
$$
\widetilde{{T_{in}}} = \sum\limits_{i = 0}^{N - 2} {{T_{i{n_i}}}}  =  - \frac{{r\left( {{V_s} - {V_T}} \right)}}{{2\pi {V_T}{V_s}}} + \frac{{{R_0}}}{{{V_s}}}\\
 - {\left( {1 + \frac{{2\pi {V_T}}}{{{V_s} + {V_T}}}} \right)^{N - 1}}\left( {\frac{{2\pi {R_0}{V_T} - r\left( {{V_s} - {V_T}} \right)}}{{2\pi {V_T}{V_s}}}} \right)
\eqno{(137)}
$$
The last inward advancement is given by,
$$
{T_{_{in}last}} = \frac{{{R_N}}}{{{V_s}}}
\eqno{(138)}
$$
We have that,
$$
{R_N} = \frac{{{c_1}}}{{1 - {c_3}}} + {c_3}^N\left( {{R_0} - \frac{{{c_1}}}{{1 - {c_3}}}} \right)
\eqno{(139)}
$$
Therefore,
$$
{R_N} = \frac{{r\left( {{V_s} - {V_T}} \right)}}{{2\pi {V_T}}} + {\left( {1 + \frac{{2\pi {V_T}}}{{{V_s} + {V_T}}}} \right)^N}\left( {\frac{{2\pi {R_0}{V_T} - r\left( {{V_s} - {V_T}} \right)}}{{2\pi {V_T}}}} \right)
\eqno{(140)}
$$
Substitution of teems yields,
$$
{T_{_{in}last}} = \frac{{r\left( {{V_s} - {V_T}} \right)}}{{2\pi {V_T}{V_s}}} + {\left( {1 + \frac{{2\pi {V_T}}}{{{V_s} + {V_T}}}} \right)^N}\left( {\frac{{2\pi {R_0}{V_T} - r\left( {{V_s} - {V_T}} \right)}}{{2\pi {V_T}{V_s}}}} \right)
\eqno{(141)}
$$
The total inwards advancement times is therefore given by,
$$
{T_{in}} = \frac{{{R_0}}}{{{V_s}}} + {\left( {1 + \frac{{2\pi {V_T}}}{{{V_s} + {V_T}}}} \right)^{N - 1}}\left( {\frac{{2\pi {R_0}{V_T} - r\left( {{V_s} - {V_T}} \right)}}{{{V_s}\left( {{V_s} + {V_T}} \right)}}} \right)
\eqno{(142)}
$$
The time it takes the sweeper formation to perform the circular portions of the sweep process is given by,
$$
{T_{circular}} = \frac{{{T_0} - {c_3}{T_{N - 1}} + \left( {N - 1} \right){c_4}}}{{1 - {c_3}}}
\eqno{(143)}
$$
Where the coefficient $c_4$ is given by,
$$
{c_4} =  - \frac{{2\pi r\left( {{V_s} - {V_T}} \right)}}{{{V_s}\left( {{V_s} + {V_T}} \right)}}
\eqno{(144)}
$$
The time it takes the sweeper formation to perform the first sweep is given by,
$$
{T_0} = \frac{{2\pi {R_0}}}{{{V_s}}}
\eqno{(145)}
$$
The time it takes the sweeper formation to perform the last circular sweep is given by,
$$
{T_{N - 1}} = \frac{{r\left( {{V_s} - {V_T}} \right)}}{{{V_s}{V_T}}} + {\left( {1 + \frac{{2\pi {V_T}}}{{{V_s} + {V_T}}}} \right)^{N - 1}}\left( {\frac{{2\pi {R_0}{V_T} - r\left( {{V_s} - {V_T}} \right)}}{{{V_s}{V_T}}}} \right)
\eqno{(146)}
$$
Therefore ${T_{circular}}$ is given by,
$$
{T_{circular}} =  - \frac{{{R_0}\left( {{V_s} + {V_T}} \right)}}{{{V_s}{V_T}}} + \frac{{r\left( {{V_s} - {V_T}} \right)\left( {{V_s} + {V_T} + 2\pi {V_T}N} \right)}}{{2\pi {V_s}{V_T}^2}}\\
 + {\left( {1 + \frac{{2\pi {V_T}}}{{{V_s} + {V_T}}}} \right)^N}\left( {\frac{{2\pi {R_0}{V_T} - r\left( {{V_s} - {V_T}} \right)}}{{{V_s}{V_T}}}} \right)\left( {\frac{{{V_s} + {V_T}}}{{2\pi {V_T}}}} \right)
\eqno{(147)}
$$
The last circular sweep occurs when the lowest tip of the sweeper's sensor is located at the center of the evader region. It is given by,
$$
{T_{last}} = \frac{{2\pi r}}{{{V_s}}}
\eqno{(148)}
$$
Therefore the total circular traversal times are given by,
$$
\begin{array}{l}
{T_{circular}} =  - \frac{{{R_0}\left( {{V_s} + {V_T}} \right)}}{{{V_s}{V_T}}} + \frac{{r\left( {{V_s} - {V_T}} \right)\left( {{V_s} + {V_T} + 2\pi {V_T}N} \right)}}{{2\pi {V_s}{V_T}^2}}\\
 + {\left( {1 + \frac{{2\pi {V_T}}}{{{V_s} + {V_T}}}} \right)^N}\left( {\frac{{2\pi {R_0}{V_T} - r\left( {{V_s} - {V_T}} \right)}}{{{V_s}{V_T}}}} \right)\left( {\frac{{{V_s} + {V_T}}}{{2\pi {V_T}}}} \right) + \frac{{2\pi r}}{{{V_s}}}
\end{array}
\eqno{(149)}
$$
\section{}
This appendix deals with calculating an analytical solution for the critical velocity. $f'(t)$ is an increasing function. At its zero crossing point the function turns from descending to increasing. We are looking for the time at which this zero crossing occurs. In order to solve for the time $t^*(V_s)$ that obtains the minimum of the function $f(t)$ we develop the second order Taylor approximation for $\sin \left( {\frac{{{V_s}t}}{R_0}} \right)$ around the point $t=0$.
\begin{thm}
The function $f(t,V_s)$ reaches its approximated minimum at time $t^*(V_s)$, where $t^*(V_s)$ is given by,
$$
{t^*(V_s)} = \frac{{2\pi {R_0}^2{V_T}^2}}{{{V_S}\left[ {{V_s}^2({R_0} + r) - {V_T}^2{R_0}} \right]}}
\eqno{(150)}
$$
\end{thm}
The complete proof can be found in Appendix $N$. After obtaining the point, $t^*$, where $f(t,V_s)$ is minimal for all values of $V_s$ we wish to find the minimal $V_s$ in which $(11)$ is satisfied, that is we wish to find the value of $V_s$ in which $f(t^*,V_s) = 0 $. Since in $(10)$ $V_s$ appears inside the argument of a cosine function as well as in a polynomial form some approximation needs to be done in order to solve for $V_s$, the approximation is valid since $\frac{{{V_S}t}}{{{R_0}}} \ll 1$. We denote the sweeper velocity obtained after the mentioned approximation as $V_{{s_2}}$. $V_{{s_2}}$ results in a value of $f(t^*,V_{s_2})$ which is slightly greater than $0$ for all choices of parameters. After developing $V_{{s_2}}$ we can also apply a bisection method between $V_{s_{\hspace{1mm} 1 \hspace{1mm} cycle}}$ and $V_{{s_2}}$ and obtain a velocity that will result in a value of $f(t^*,V_{c})$ that is close to $0$ with any desired accuracy.
\begin{thm}
In order to find the minimal $V_s$ in which $f(t^*,V_s) = 0$ we obtain another velocity named $V_{{s_2}}$  that results in a value of $f(t^*,V_{s_2})$ which is slightly greater than $0$ for all choices of parameters. This velocity is derived from the demand that the spread of the evader region after one traversal of the scanning agent will be $r - {t^*}{V_T}$ instead of $r$. It is given by,
$$
{{V_{{s_2}}} = \frac{{\pi {R_0}{V_T}\left( {{R_0} + r} \right) + {V_T}\sqrt {{R_0}\left( {{R_0} + r} \right)\left[ {{\pi ^2}{R_0}\left( {{R_0} + r} \right) + {r^2}} \right]} }}{{r\left( {{R_0} + r} \right)}}}
\eqno{(151)}
$$
\end{thm}
The full derivation of ${V_{{s_2}}}$ is provided in Appendix $O$. Plugging this value of ${V_{{s_2}}} = V_c$ yields that $f(t,V_c) \geq 0 \hspace{1mm} \forall t \in [0,\frac{\pi R_0}{2V_c}]$.
The calculation of the critical velocity was verified for the following parameters:
$R_0 = 100, r=10, V_T = 1$. The new critical velocity that achieved is ${V_{{s_2}}} = 62.84631837$. The previously calculated critical velocity that was based on only considering one cycle and was calculated as  $V_{s_{\hspace{1mm} 1 \hspace{1mm} cycle}} = \frac{{2\pi {R_0}{V_T}}}{r} = 62.83185307$.
In Fig. $16$ a plot of $f(t,{V_{{s_2}}})$ is presented. It can now be observed that for all relevant times this inequality is positive, and hence ensures a guaranteed no escape searcher velocity for all possible evader trajectories satisfying a maximal evader velocity of $V_T$.
\begin{figure}[ht]
\noindent \centering{}\includegraphics[width=3.5in,height =2.8in]{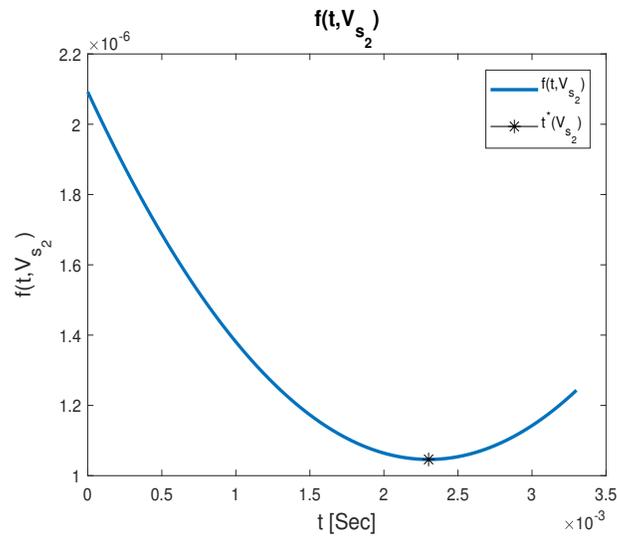} \caption{$f(t,V_{s_2}$ for the new critical velocity $V_c= {V_{{s_2}}}$. It can now be observed that for all relevant times this inequality is positive, and hence ensures a guaranteed no escape searcher velocity for all possible evader trajectories satisfying a maximal evader velocity of $V_T$.The parameters values chosen for this plot are $r=10$, $V_T = 1$ and $R_0 = 100$.}
\label{Fig16Label}
\end{figure}

\begin{figure}[ht]
\noindent \centering{}\includegraphics[width=3.5in,height =3in]{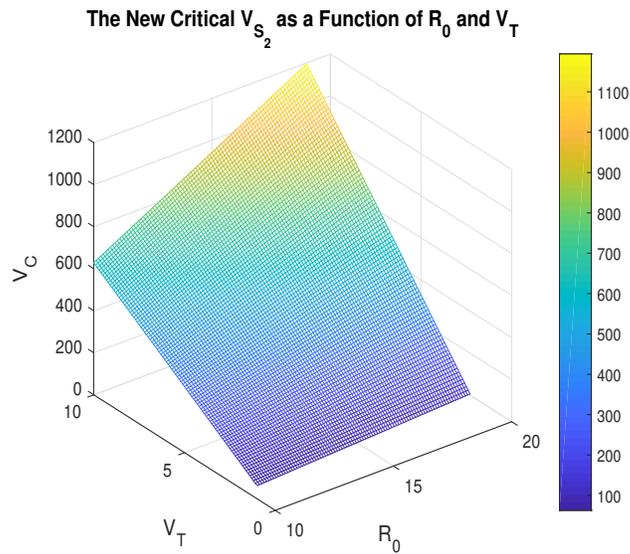} \caption{Mesh plot of $V_c= {V_{{s_2}}}$ as a function of $R_0$ and $V_T$. $V_T$ is the maximal evader velocity.In this plot half the length of the sensor was taken as $r=1$.}
\label{Fig17Label}
\end{figure}
\begin{thm}
For all search parameters satisfying that $R_0\geq r$ it holds that  $f(t^*,V_{s_2}) > 0$.  Therefore it holds that $\forall \hspace{1 mm} t ,\hspace{1 mm} f(t^*,V_{s_2}) \geq 0$. Thus $V_{s_2}$ is a sufficiently high velocity in order to accomplish the confinement task.
\end{thm}
\begin{figure}[ht]
\noindent \centering{}\includegraphics[width=3.5in,height =3in]{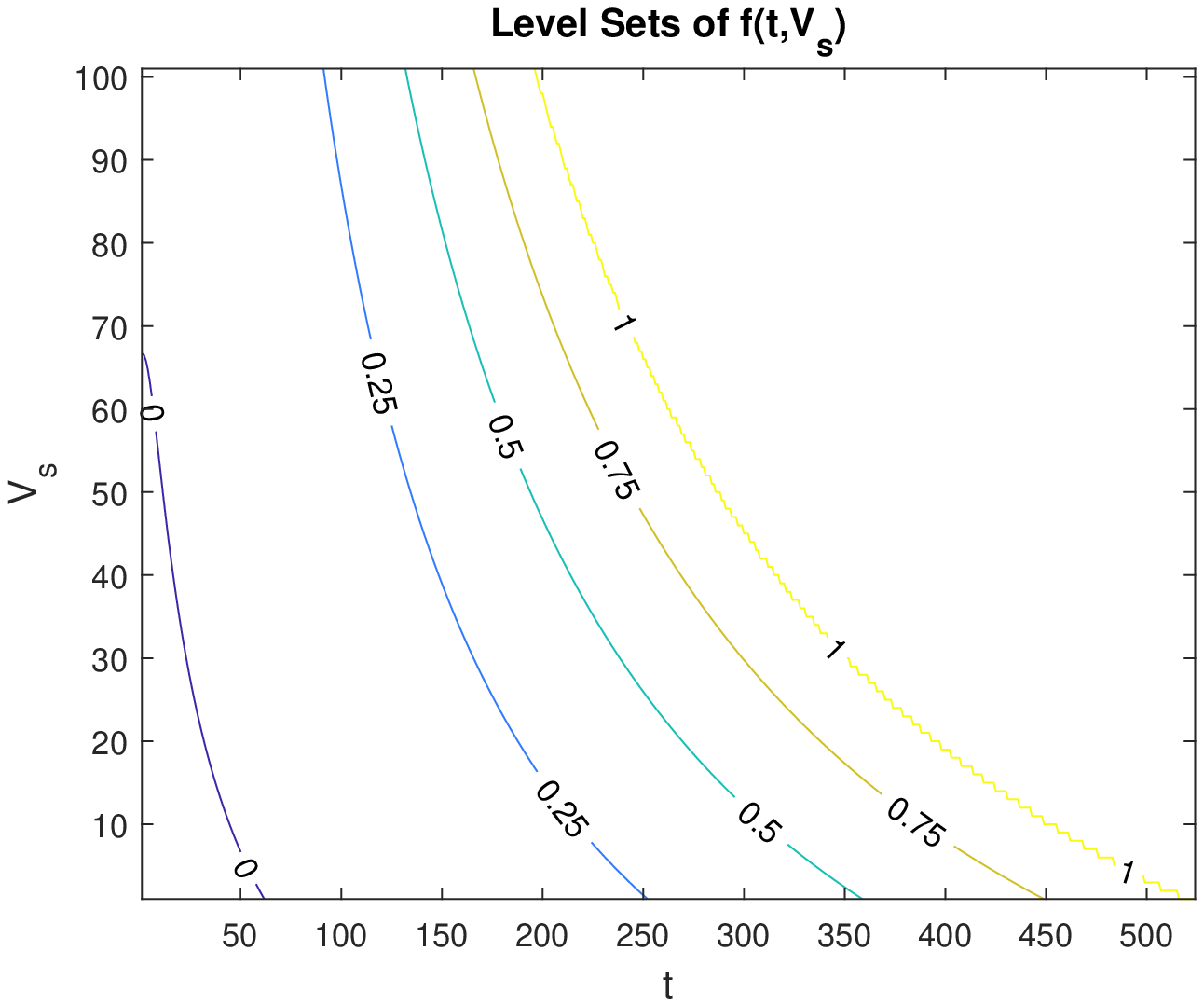} \caption{Level Sets of $f(t,V_{s_2})$. The zero crossings are the critical values of $V_s$ we are looking for. In order to guarantee a non escape velocity the searcher velocity has to be in the region to the right of the zero level set.}
\label{Fig18Label}
\end{figure}

\begin{figure}[ht]
\noindent \centering{}\includegraphics[width=3.2in,height =3.2in]{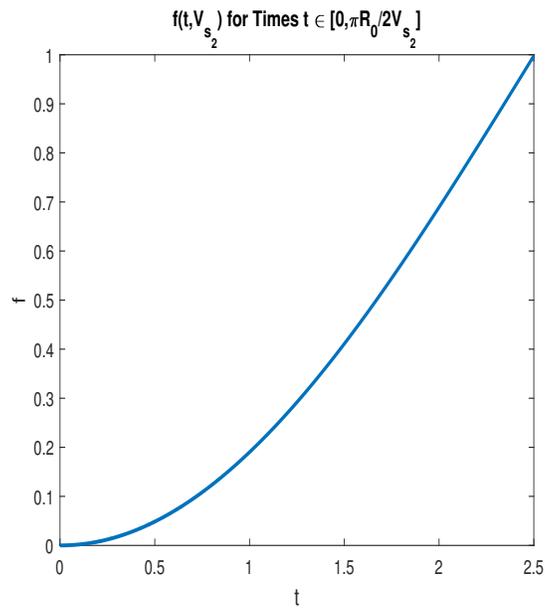} \caption{$f(t,V_{s_{2}})$ for $t \in [0,\frac{{\pi {R_0}}}{{2{V_{s_{2}}}}}]$. It can be seen that for all times of interest $f(t,V_{s_{2}}) \ge 0$. The parameters values chosen for this plot are $r=10$, $V_T = 1$ and $R_0 = 100$.}
\label{Fig19Label}
\end{figure}

\begin{figure}[ht]
\noindent \centering{}\includegraphics[width=3.8in,height =4in]{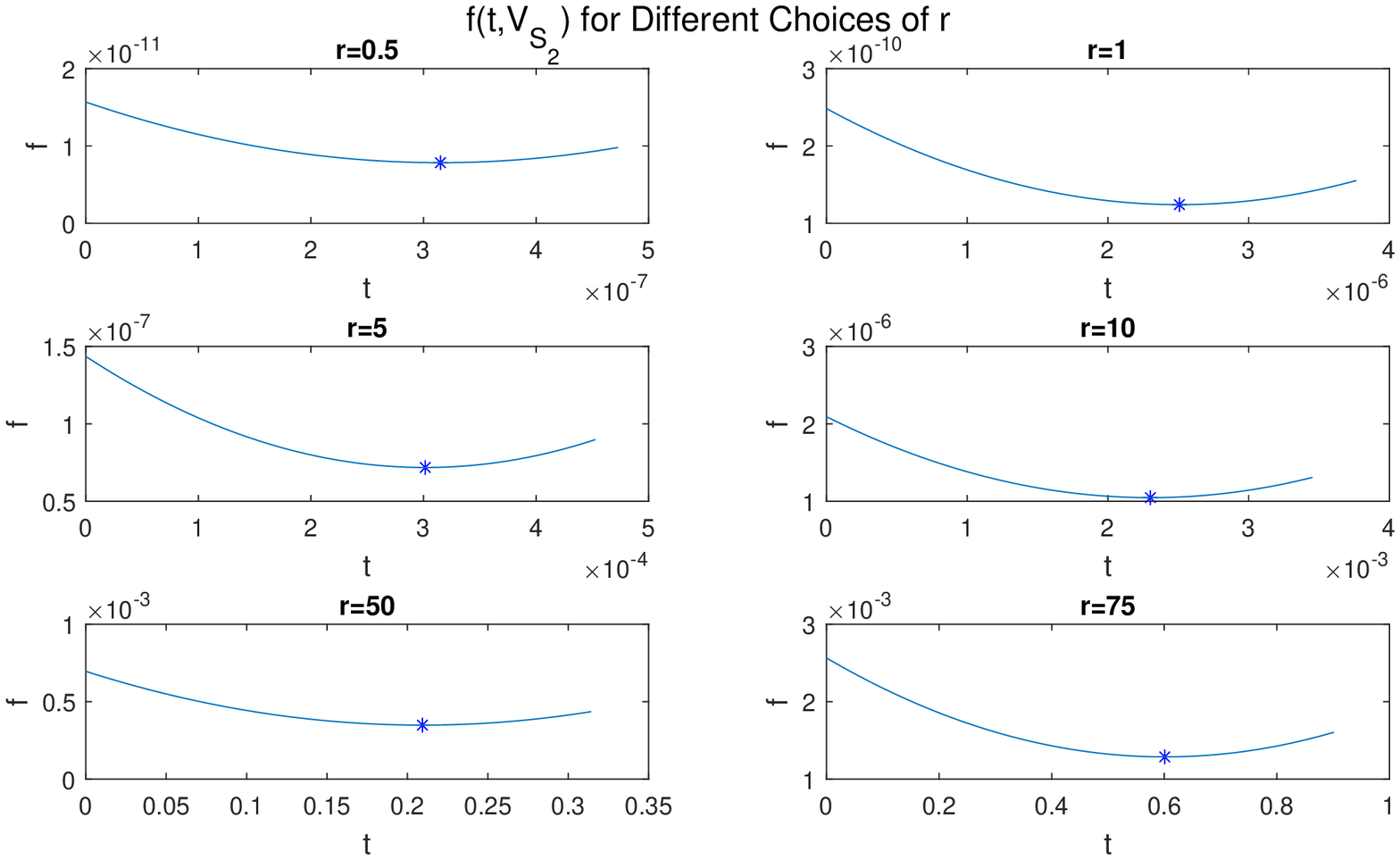} \caption{$f(t,V_{s_{2}})$ for different choices of $r$ and for a given initial evader region radius of $R_0 =100$. The blue asterisk in each plot denotes the estimated minimum of $f(t,V_{s_{2}})$ which occurs at $t=t^*$. The parameters values chosen for this plot are $V_T = 1$ and $R_0 = 100$.}
\label{Fig20Label}
\end{figure}

For proof that $V_c$ is greater than ${V_{{s_2}}}$ see Appendix $P$. This implies that since ${V_c}$ is greater than ${V_{{s_2}}}$ and both ${V_{{s_2}}}$ and ${V_c}$ ensure the satisfaction of the confinement task, choosing ${V_c}$ as the critical velocity results in a less tight satisfaction of the task.

\section{}
\textbf{Theorem 9. }
The function $f(t,V_s)$ reaches its approximated minimum at time $t^*(V_s)$, where $t^*(V_s)$ is given by,
$$
{t^*(V_s)} = \frac{{2\pi {R_0}^2{V_T}^2}}{{{V_S}\left[ {{V_s}^2({R_0} + r) - {V_T}^2{R_0}} \right]}}
\eqno{(152)}
$$
\begin{proof}
We have that,
$$
f'(t) =  - \frac{{{V_T}^2}}{{{R_0}({R_0} + r)}}\left( {\frac{{2\pi {R_0}}}{{{V_s}}} + t} \right) + \sin \left( {\frac{{{V_s}t}}{{{R_0}}}} \right)\frac{{{V_s}}}{{{R_0}}}
\eqno{(153)}
$$
Simplifying notations and equating to zero in order to find $t^*(V_s)$ results in,
$$
f'(t) =  - \frac{{2\pi {V_T}^2}}{{({R_0} + r){V_s}}} - \frac{{{V_T}^2t}}{{{R_0}({R_0} + r)}}
 + \sin \left( {\frac{{{V_s}t}}{{{R_0}}}} \right)\frac{{{V_s}}}{{{R_0}}} \equiv 0
 \eqno{(154)}
$$
$f'(t)$ is an increasing function. At its zero crossing point the function turns from descending to increasing. we are looking for the time at which this zero crossing occurs. In order to solve for the time $t^*$ that obtains the minimum of the function $f(t)$ we will develop the second order Taylor approximation for $\sin \left( {\frac{{{V_s}t}}{R_0}} \right)$ around the point $t=0$. The Taylor approximation for $\sin \left( {\frac{{{V_s}t}}{R_0}} \right)$ around the point $t=0$ is,
$$
\sin \left( {\frac{{{V_s}t}}{R}} \right) = \sum\limits_{n = 0}^\infty  {\frac{{{{( - 1)}^n}}}{{(2n + 1)!}}} {\left( {\frac{{{V_s}t}}{R_0}} \right)^{2n + 1}}
\eqno{(155)}
$$
By rearranging terms and equating to zero in order to find $t^*(V_s)$  we have the following,
$$
f'(t) =  - \frac{{2\pi {V_T}^2}}{{({R_0} + r){V_s}}} - \frac{{{V_T}^2t}}{{{R_0}({R_0} + r)}}
 + \sin \left( {\frac{{{V_s}t}}{{{R_0}}}} \right)\frac{{{V_s}}}{{{R_0}}} \equiv 0
\eqno{(156)}
$$
By plugging the Taylor series expression into $(156)$ we get that,
$$
\frac{{ - 2\pi {V_T}^2}}{{({R_0} + r){V_s}}} - \frac{{{V_T}^2t}}{{{R_0}({R_0} + r)}} +
\frac{{{V_s}}}{{{R_0}}}\sum\limits_{n = 0}^\infty  {\frac{{{{( - 1)}^n}}}{{(2n + 1)!}}} {\left( {\frac{{{V_s}}}{R}} \right)^{2n + 1}}{t^{2n + 1}}
\eqno{(157)}
$$
And by rearranging terms to,
$$
t\left( {\frac{{{V_s}^2}}{{{R_0}^2}} - \frac{{{V_T}^2}}{{{R_0}({R_0} + r)}}} \right) +
\sum\limits_{n = 1}^\infty  {\frac{{{{( - 1)}^n}}}{{(2n + 1)!}}} {\left( {\frac{{{V_s}}}{{{R_0}}}} \right)^{2n + 2}}{t^{2n + 1}} = \frac{{2\pi {V_T}^2}}{{({R_0} + r){V_s}}}
\eqno{(158)}
$$
Since analytical solutions to polynomial functions exist only for polynomials of a degree less than or equal to four we will develop only the first order approximation of the Taylor series since this will result at the end of the calculations in a fourth order polynomial function of $V_s$. Adding higher order terms of the Taylor series to the expression will result in polynomials of order greater than $4$ of $V_s$. According to the Abel-Ruffini theorem that states that there is no algebraic solution to the general polynomial equations of higher degrees than $4$ adding such terms will prevent developing an analytic expression to $V_s$ and therefore are not presented in this work. Therefore $(158)$ resolves under the first order approximation to,
$$
t\left( {\frac{{{V_s}^2}}{{{R_0}^2}} - \frac{{{V_T}^2}}{{{R_0}({R_0} + r)}}} \right) = \frac{{2\pi {V_T}^2}}{{({R_0} + r){V_s}}}
\eqno{(159)}
$$
Rearranging terms results in,
$$
t\left( {\frac{{{V_s}^2({R_0} + r) - {V_T}^2{R_0}}}{{{R_0}^2({R_0} + r)}}} \right) = \frac{{2\pi {V_T}^2}}{{({R_0} + r){V_s}}}
\eqno{(160)}
$$
And finally we get that the function $f(t)$ attains its minimum at the point, which we will denote as $t^*(V_s)$ at,
$$
{t^*(V_s)} = \frac{{2\pi {R_0}^2{V_T}^2}}{{{V_s}\left[ {{V_s}^2({R_0} + r) - {V_T}^2{R_0}} \right]}}
\eqno{(161)}
$$
$\openbox$.
\end{proof}
\section{}
\textbf{Theorem 10. }
In order to find the minimal $V_s$ in which $f(t^*,V_s) = 0$ we obtain another velocity named $V_{{s_2}}$  that results in a value of $f(t^*,V_{s_2})$ which is slightly greater than $0$ for all choices of parameters. This velocity is derived from the demand that the spread of the evader region after one traversal of the scanning agent will be $r - {t^*}{V_T}$ instead of $r$. It is given by,
$$
{{V_{{s_2}}} = \frac{{\pi {R_0}{V_T}\left( {{R_0} + r} \right) + {V_T}\sqrt {{R_0}\left( {{R_0} + r} \right)\left[ {{\pi ^2}{R_0}\left( {{R_0} + r} \right) + {r^2}} \right]} }}{{r\left( {{R_0} + r} \right)}}}
\eqno{(162)}
$$
\begin{proof}
A second method that results in an upper bound on the approximation of $V_s$ is obtained
by plugging the obtained value of $t^*$ into the following equations to obtain the critical velocity,
$$
\frac{{2\pi {R_0}}}{{{V_S}}} + {t^*} < \frac{r}{{{V_T}}}
\eqno{(163)}
$$

$$
\frac{{2\pi {R_0}}}{{{V_S}}} < \frac{{r - {t^*}{V_T}}}{{{V_T}}}
\eqno{(164)}
$$

$$
{V_S} > \frac{{2\pi {R_0}{V_T}}}{{r - {t^*}{V_T}}}
\eqno{(165)}
$$
This method is equivalent to a demand that the spread of the evader region after one traversal of the scanning agent will be $r - {t^*}{V_T}$ instead of $r$. After plugging $t^*$ into $(165)$ we have,
$$
{V_s} \geq \frac{{2\pi {R_0}{V_T}{V_S}\left[ {{V_S}^2\left( {{R_0} + r} \right) - {V_T}^2{R_0}} \right]}}{{r{V_S}\left[ {{V_S}^2\left( {{R_0} + r} \right) - {V_T}^2{R_0}} \right] - {V_T}^32\pi {R_0}^2}}
\eqno{(166)}
$$
Grouping of the elements with the same power of $V_s$ and rearranging terms yields,
$$
{V_S}^3r\left( {{R_0} + r} \right) - {V_S}^22\pi {R_0}{V_T}\left( {{R_0} + r} \right) - {V_S}r{V_T}^2{R_0}
\eqno{(167)}
$$
Since the solution of $V_s=0$ means that the agent will stay put in its place for all times and therefore there will no escape from the most problematic point is the trivial solution that does not interest us so it can be discarded. Therefore we obtain a quadratic equation for $V_s$ and choose the maximal solution that corresponds for a full traversal of the evader region. The second smaller result that is obtained means that the agent will move very slow and catch the most problematic point but will not make a complete cycle around the evader region. Therefore the solution for the critical velocity is,
$$
{{V_{{s_2}}} = \frac{{\pi {R_0}{V_T}\left( {{R_0} + r} \right) + {V_T}\sqrt {{R_0}\left( {{R_0} + r} \right)\left[ {{\pi ^2}{R_0}\left( {{R_0} + r} \right) + {r^2}} \right]} }}{{r\left( {{R_0} + r} \right)}}}
\eqno{(168)}
$$
$\openbox$.
\end{proof}
\section{}
\textbf{Theorem 11. }
For all search parameters satisfying that $R_0\geq r$ it holds that  $f(t^*,V_{s_2}) > 0$.  Therefore it holds that $\forall \hspace{1 mm} t ,\hspace{1 mm} f(t^*,V_{s_2}) \geq 0$. Thus $V_{s_2}$ is a sufficiently high velocity in order to accomplish the confinement task.
\begin{proof}
We first assume  that $R_0$ can be expressed as $R_0 = \alpha r$ , where $\alpha  \geq 1$. Without loss of generality we assume that $r=1$ and $V_T =1$. We can assume these values since all the other search parameters depend on them. In section $5$ we saw that ${V_{{s_2}}}$ can be expressed as,
$$
{V_{{s_2}}} = \frac{{\pi {R_0}{V_T}\left( {{R_0} + r} \right) + {V_T}\sqrt {{R_0}\left( {{R_0} + r} \right)\left[ {{\pi ^2}{R_0}\left( {{R_0} + r} \right) + {r^2}} \right]} }}{{r\left( {{R_0} + r} \right)}}
\eqno{(169)}
$$
Substituting $R_0 = \alpha r$, $r=1$ and $V_T =1$ yields that,
$$
{V_{{s_2}}} = \frac{{\alpha \pi \left( {\alpha  + 1} \right) + \sqrt {\alpha \left( {\alpha  + 1} \right)\left[ {{\pi ^2}\alpha \left( {\alpha  + 1} \right) + 1} \right]} }}{{\alpha  + 1}}
\eqno{(170)}
$$
${V_{{s_2}}}$ is a monotonically increasing function in $\alpha$.
$$
f({t^*},{V_{{s_2}}}) = 1 + \frac{1}{{2\alpha (\alpha  + 1)}}\left( {1 - {{\left( {\frac{{2\pi \alpha }}{{{V_{{s_2}}}}} + {t^*}} \right)}^2}} \right) - \cos \left( {\frac{{{V_S}{t^*}}}{\alpha }} \right)
\eqno{(171)}
$$
$f(t,{V_{{s_2}}})$ is monotonically decreasing function in $\alpha$. The argument inside the cosine function decays to $0$ as alpha increases since $t^*$ goes to as alpha increases. The second argument in $(171)$ also tends to $0$ as $\alpha$ increases. This behaviour in plotted in Fig. $21$. Therefore in order to lower bound $f(t^*,{V_{{s_2}}})$ we will prove that,
$$
\mathop {\lim }\limits_{\alpha  \to \infty } f({t^*},{{V_{{s_2}}}}) \ge 0
\eqno{(172)}
$$
since $t^*$ is the point in time in which $f(t,{V_{{s_2}}})$ is minimal, ensuring that $f({t^*},{V_{{s_2}}}) \ge 0$ guarantees that $f({t},{V_{{s_2}}}) \ge 0 \forall \hspace{1 mm} t $. Letting $\alpha \to \infty$ yields,
$$
\mathop {\lim }\limits_{\alpha  \to \infty } {V_{{s_2}}} = \frac{{{\alpha ^2}\pi  + \sqrt {{\alpha ^4}{\pi ^2}} }}{\alpha }
\eqno{(173)}
$$
We therefore obtain that,
$$
\mathop {\lim }\limits_{\alpha  \to \infty } {V_{{s_2}}} = \mathop {\lim }\limits_{\alpha  \to \infty } 2\alpha \pi
\eqno{(174)}
$$
In theorem $9$ we proved that the function $f(t,V_s)$ reaches its minimum at time $t^*(V_s)$, where $t^*(V_s)$ is given by,
$$
{t^*}({V_{{s_2}}}) = \frac{{2\pi {R_0}^2{V_T}^2}}{{{V_{{s_2}}}\left[ {{V_{{s_2}}}^2({R_0} + r) - {V_T}^2{R_0}} \right]}}
\eqno{(175)}
$$
Plugging into $(175)$ the chosen parameters yields,
$$
{t^*}({V_{{s_2}}}) = \frac{{2\pi {\alpha ^2}}}{{{V_{{s_2}}}\left[ {{V_{{s_2}}}^2(\alpha  + 1) - \alpha } \right]}}
\eqno{(176)}
$$
Letting $\alpha \to \infty$ yields,
$$
\mathop {\lim }\limits_{\alpha  \to \infty } {t^*}({V_{{s_2}}}) = \mathop {\lim }\limits_{\alpha  \to \infty } \frac{{2\pi {\alpha ^2}}}{{8{\alpha ^4}{\pi ^3}}} = \mathop {\lim }\limits_{\alpha  \to \infty } \frac{1}{{4{\alpha ^2}{\pi ^2}}} = 0
\eqno{(177)}
$$
We therefore have that,
$$
f({t^*},{V_{{s_2}}}) = 1 + \frac{1}{{2\alpha (\alpha  + 1)}}\left( {1 - {{\left( {\frac{{2\pi \alpha }}{{{V_{{s_2}}}}} + {t^*}} \right)}^2}} \right) - \cos \left( {\frac{{{V_{{s_2}}}{t^*}}}{\alpha }} \right)
\eqno{(178)}
$$
Letting $\alpha \to \infty$ yields,
$$
\mathop {\lim }\limits_{\alpha  \to \infty } {V_{{s_2}}} = \mathop {\lim }\limits_{\alpha  \to \infty } 2\alpha \pi
\eqno{(179)}
$$
We therefore obtain that,
$$
\mathop {\lim }\limits_{\alpha  \to \infty } \frac{{{V_{{s_2}}}{t^*}}}{\alpha } = \mathop {\lim }\limits_{\alpha  \to \infty } 2\pi {t^*} \to 0
\eqno{(180)}
$$
Therefore we have that ${t^*}$ satisfies,
$$
\mathop {\lim }\limits_{\alpha  \to \infty } f({t^*},{V_{{s_2}}}) = 0
\eqno{(181)}
$$
This results in the desired behaviour that guarantees that,
$$
f(t,{V_{{s_2}}}) \ge 0,\forall \hspace{1 mm} t, \alpha \geq 1
\eqno{(182)}
$$
$\openbox$.
\end{proof}

\begin{figure}[ht]
\noindent \centering{}\includegraphics[width=3in,height =2.5in]{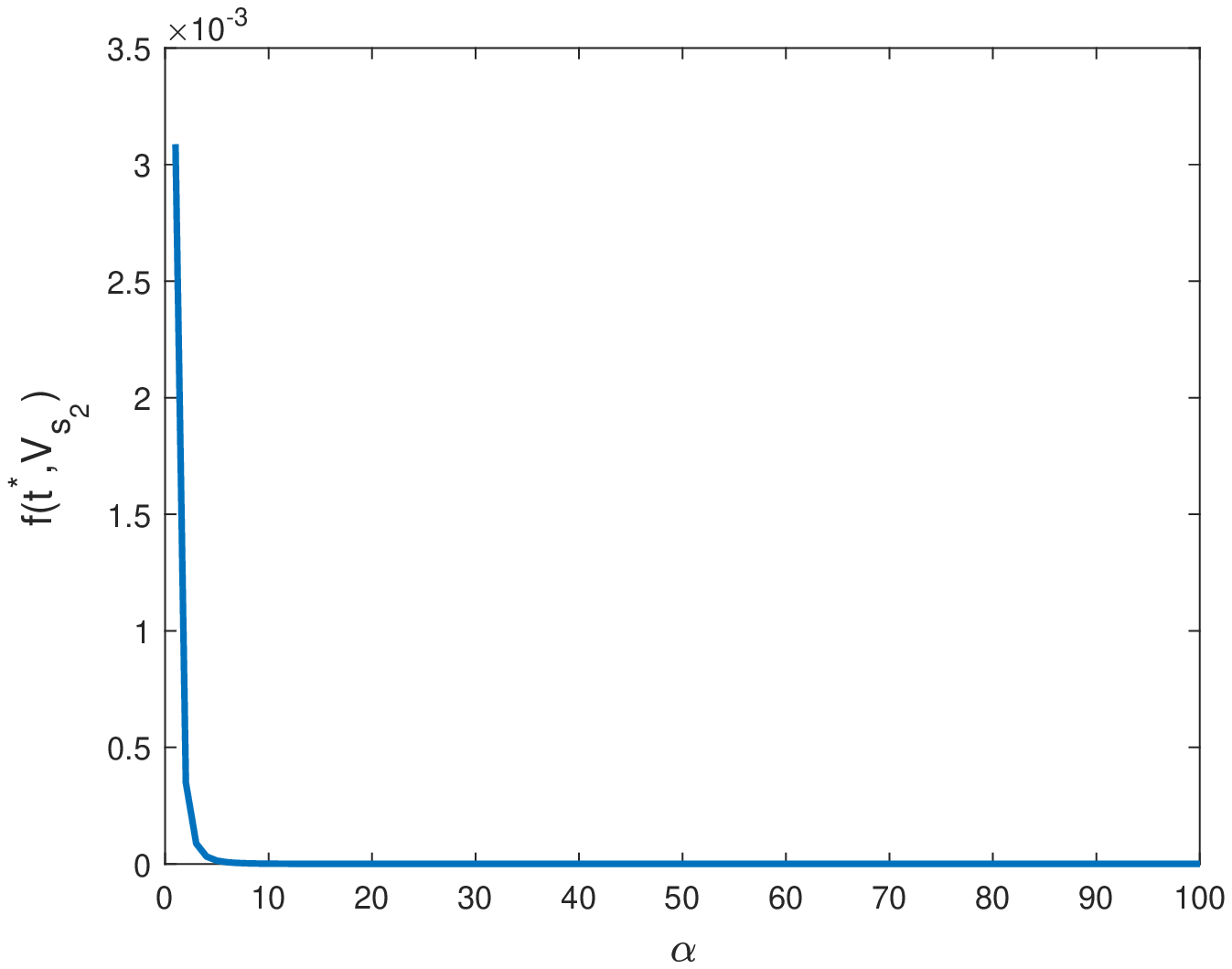} \caption{$f(t^*(V_{s_2}),V_{s_2})$ as a function of $\alpha$. For this plot $V_T =1$, $r=1$.}
\label{Fig21Label}
\end{figure}

\begin{figure}[ht]
\noindent \centering{}\includegraphics[width=3in,height =2.5in]{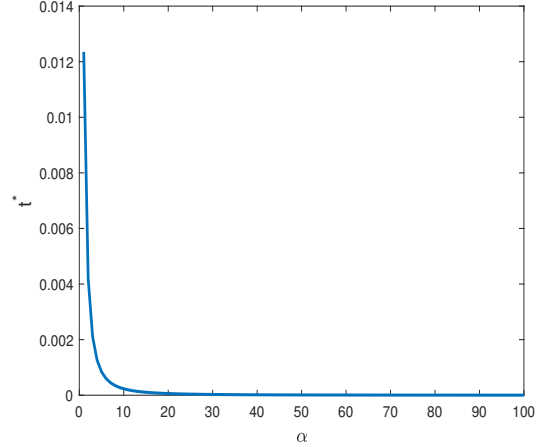} \caption{$t^*(V_{s_2})$ as a function of $\alpha$. For this plot $V_T =1$, $r=1$.}
\label{Fig22Label}
\end{figure}

\section{}
\newtheorem{lemma}{Lemma}
\begin{lemma}
The minimal critical velocity ${V_C}_{\_new}$ is greater than the upper bound on the critical velocity that was derived for $f(t,V_s)$,${V_{{s_2}}}$ by
$$
{V_{{C_{new}}}} - {V_{{s_2}}} = \frac{{2{V_T}r\left[ {2\pi {R_0}\left( {{R_0} + r} \right) + {r^2}} \right]}}{{\pi {R_0} + 2r + \sqrt {{R_0}\left( {{R_0} + r} \right)\left[ {{\pi ^2}{R_0}\left( {{R_0} + r} \right) + {r^2}} \right]} }} > 0
\eqno{(183)}
$$
\end{lemma}
\begin{proof}
We have that ${V_{{C_{new}}}}$ satisfies,
$$
{V_{{C_{new}}}} = \frac{{2\pi {R_0}{V_T}}}{r} + {V_T}
\eqno{(184)}
$$
And $V_{{s_2}}$ the lower critical velocity that was derived from $f(t,V_s)$ satisfies,
$$
{{V_{{s_2}}} = \frac{{2\pi {R_0}{V_T}\left( {{R_0} + r} \right) + 2{V_T}\sqrt {{R_0}\left( {{R_0} + r} \right)\left[ {{\pi ^2}{R_0}\left( {{R_0} + r} \right) + {r^2}} \right]} }}{{2r\left( {{R_0} + r} \right)}}}
\eqno{(185)}
$$
We need to prove that ${V_{{C_{new}}}} - V_{{s_2}} > 0$. Rewriting ${V_{{C_{new}}}}$ we have,
$$
{{V_{{C_{new}}}} = \frac{{2\pi {R_0}{V_T} + {V_T}r}}{r} = \frac{{4\pi {R_0}{V_T}\left( {{R_0} + r} \right) + 2{V_T}r\left( {{R_0} + r} \right)}}{{2r\left( {{R_0} + r} \right)}}}
\eqno{(186)}
$$
And the difference between the two critical velocities is given by,
$$
{{V_{{C_{new}}}} - {V_{{s_2}}} = \frac{{2\pi {R_0}{V_T}\left( {{R_0} + r} \right) + 2{V_T}r\left( {{R_0} + r} \right) - 2{V_T}\sqrt {{R_0}\left( {{R_0} + r} \right)\left[ {{\pi ^2}{R_0}\left( {{R_0} + r} \right) + {r^2}} \right]} }}{{2r\left( {{R_0} + r} \right)}}}
\eqno{(187)}
$$
Multiplying and dividing the denominator and numerator by,
$$
2\pi {R_0}{V_T}\left( {{R_0} + r} \right) + 2{V_T}r\left( {{R_0} + r} \right) +
2{V_T}\sqrt {{R_0}\left( {{R_0} + r} \right)\left[ {{\pi ^2}{R_0}\left( {{R_0} + r} \right) + {r^2}} \right]}
\eqno{(188)}
$$
yields,
$$
{{V_{{C_{new}}}} - {V_{{s_2}}} = \frac{{{{\left( {\left( {{R_0} + r} \right)2{V_T}\left[ {\pi {R_0} + r} \right]} \right)}^2} - 4{V_T}^2\left( {{{\left( {{R_0} + r} \right)}^2}{\pi ^2}{R_0}^2 + {R_0}\left( {{R_0} + r} \right){r^2}} \right)}}{{2r\left( {{R_0} + r} \right)\left( {2\pi {R_0}{V_T}\left( {{R_0} + r} \right) + 2{V_T}r\left( {{R_0} + r} \right) + 2{V_T}\sqrt {{R_0}\left( {{R_0} + r} \right)\left[ {{\pi ^2}{R_0}\left( {{R_0} + r} \right) + {r^2}} \right]} } \right)}}}
\eqno{(189)}
$$
Which reduces to,
$$
{V_{{C_{new}}}} - {V_{{s_2}}} = \frac{{2{V_T}r\left[ {2\pi {R_0}\left( {{R_0} + r} \right) + {r^2}} \right]}}{{\pi {R_0} + 2r + \sqrt {{R_0}\left( {{R_0} + r} \right)\left[ {{\pi ^2}{R_0}\left( {{R_0} + r} \right) + {r^2}} \right]} }} > 0
\eqno{(190)}
$$
Since all the terms are positive and the expression contains only addition signs the expression is obviously positive and hence the Lemma holds. Therefore we can use ${V_{{C_{new}}}}$ as the critical velocity and obtain guaranteed coverage.$\openbox$.
\end{proof}
For a choice of the following parameters: $R_0 = 100, r=10, V_T = 1$, the difference between the velocities yields, ${V_{{C_{new}}}}-{V_{{s_2}}}=63.8319-62.8463=0.9855$.

\end{document}